\documentclass[12pt]{article}
\usepackage{amsmath}
\usepackage{graphicx,psfrag,epsf}
\usepackage{enumerate}
\usepackage{natbib}
\usepackage{url} % not crucial - just used below for the URL 
\usepackage{graphicx}
\usepackage{amsthm}
\usepackage{amsmath}
\usepackage{natbib}
\usepackage[colorlinks,citecolor=blue,urlcolor=blue,filecolor=blue,backref=page]{hyperref}
\usepackage{subcaption}
\usepackage{amssymb}
\usepackage{amsthm}
\usepackage{empheq}
\usepackage{enumitem}
\usepackage{algorithm}

\numberwithin{equation}{section}
\theoremstyle{plain}

\newtheorem{lemma}{Lemma}
\newtheorem{theorem}{Theorem}
\newtheorem{corollary}{Corollary}[theorem]

\DeclareMathOperator*{\argmax}{arg\,max}

\DeclareMathOperator*{\mle}{\hat{\theta}_0}
\DeclareMathOperator*{\prob}{\overset{P}{\longrightarrow}}
\DeclareMathOperator*{\ltrue}{\text{L}_{n}(\theta \mid \textit{h}_0(\textit{Z}),\phi_0)}
\DeclareMathOperator*{\lalltrue}{\text{L}_{n}(\theta_0 \mid \textit{h}_0(\textit{Z}),\phi_0)}
\DeclareMathOperator*{\lest}{\text{L}_{n}(\theta \mid \textit{h}(\textit{Z}),\phi)}
\DeclareMathOperator*{\lmle}{\text{L}_{n}(\hat{\theta}_0 \mid \textit{h}_0(\textit{Z}),\phi_0)}
\DeclareMathOperator{\logit}{\text{logit}}
\newcommand{\blind}{1}

\begin{document}

\def\spacingset#1{\renewcommand{\baselinestretch}%
{#1}\small\normalsize} \spacingset{1}

%%%%%%%%%%%%%%%%%%%%%%%%%%%%%%%%%%%%%%%%%%%%%%%%%%%%%%%%%%%%%%%%%%%%%%%%%%%%%%

\if1\blind
{
  \title{\bf Valid Bayesian Inference based on Variance Weighted Projection for
High-Dimensional Logistic Regression with Binary Covariates}
  \author{Abhishek Ojha \hspace{.2cm}\\
    Department of Statistics, University of Illinois, Urbana-Champaign\\
    and \\
    Naveen N. Narisetty \\
    Department of Statistics, University of Illinois, Urbana-Champaign}
  \maketitle
} \fi

\if0\blind
{
  \bigskip
  \bigskip
  \bigskip
  \begin{center}
    {\LARGE\bf Title}
\end{center}
  \medskip
} \fi

\bigskip
\begin{abstract}
% We address the issue of conducting inference for a binary treatment effect concerning a binary outcome variable, all the while accounting for high-dimensional baseline covariates. 
% % Within the context of a logistic regression model, the quantity of interest is known as the causal log-odds ratio. 
% To tackle this challenge, we devise an orthogonal score tailored to this context and propose a novel Bayesian framework for performing inference for the causal log-odds ratio parameter in a high-dimensional logistic regression setting. We provide uniform convergence results that show the validity of credible intervals resulting from the posterior. Our method has competitive empirical performance when compared with state-of-the-art methods.

% \textcolor{red}{Need to add some challenge that this setting possesses compared to the continuous covariate case!}

We address the challenge of conducting inference for a categorical treatment effect related to a binary outcome variable while taking into account high-dimensional baseline covariates. The conventional technique used to establish orthogonality for the treatment effect from nuisance variables in continuous cases is inapplicable in the context of binary treatment.
To overcome this obstacle, an orthogonal score tailored specifically to this scenario is formulated which is based on a variance-weighted projection. Additionally, a novel Bayesian framework is proposed to facilitate valid inference for the desired low-dimensional parameter within the complex framework of high-dimensional logistic regression. We provide uniform convergence results, affirming the validity of credible intervals derived from the posterior distribution. The effectiveness of the proposed method is demonstrated through comprehensive simulation studies and real data analysis.
\end{abstract}

\noindent%
{\it Keywords:}  3 to 6 keywords, that do not appear in the title
\vfill

\newpage
\spacingset{1.45} % DON'T change the spacing!
\section{Introduction}

\subsection{A Motivating Example}\label{motivating_eg}

Consider the case of Chronic Kidney Disease (CKD), where kidneys cannot efficiently filter blood, leading to the retention of excess fluid and waste. This condition can result in various complications such as anemia, depression, heart disease, and stroke. Notably, CKD affects more than one in seven individuals in the United States \citep{ckd}. The occurrence of CKD is influenced by a multitude of factors, including diabetes, high blood pressure, cardiovascular diseases, high cholesterol levels, exposure to toxins, age, demographics, family history, smoking habits, autoimmune diseases, and more. Particularly, diabetes is highly prevalent in the US, with over 37 million diagnosed cases and over 96 million individuals classified as prediabetic \citep{diabetes_report}.
Extensive research has established the common occurrence of CKD among individuals suffering from diabetes \citep{dm2,dm1,dm3}. Consequently, it is imperative to quantify the association between diabetes and CKD. To perform this task accurately, it is crucial to control for all the other factors. 

We examine the Chronic Kidney Disease (CKD) data sourced from the UCI repository \citep{ckd_data}. This dataset comprises records from nearly 400 patients, among whom 246 have been diagnosed with CKD. Each patient's information includes variables such as the presence of diabetes, blood pressure readings, age, presence of heart disease, anemia, blood urea levels, and more. Utilizing this dataset, one can attempt to estimate the relationship between diabetes and CKD while accounting for various other factors.

However, it is important to note that the estimate derived from this dataset does not directly represent the population-level association between diabetes and CKD. This dataset constitutes only a small fraction of individuals in the US, randomly selected. Consequently, the estimate lacks broad applicability to the entire population. Therefore, constructing valid confidence intervals for the association between diabetes and CKD becomes crucial. These intervals offer a statistically rigorous way to understand the true relationship within the larger population, providing more reliable insights for practitioners. 

One can employ a logistic regression model to investigate the relationship between the occurrence of CKD and the presence of diabetes mellitus, and other potential factors. 
The logistic model is well-studied and commonly employed in the case of binary regression \citep{kleinbaum2002logistic,collett2002modelling,agresti2012categorical,hosmer2013applied}. It provides a direct interpretation in terms of the logarithm of the odds in favor of success, which proves valuable in epidemiological and pharmacological studies. This model enables practitioners to calculate the importance of features conveniently. Additionally, it is computationally efficient, especially in high-dimensional scenarios, as regularization can be easily applied with minimal tuning of hyper-parameters. Motivated by these advantages, we explore the challenge of conducting Bayesian inference for a low-dimensional parameter of interest within a high-dimensional logistic regression model. We choose Bayesian inference to integrate any existing prior knowledge about diabetes within a specific population.
We revisit this data analysis in Section~\ref{ckd_analysis}. 

\subsection{Problem Formulation}

 Our data consist of a sample of $n$ observations containing binary responses denoted by $Y_i, i = 1,\cdots,n,$ and a binary covariate of interest denoted by $X_i, i = 1,\cdots,n,$ and $d$-dimensional nuisance variables denoted by $Z_i, i = 1,\cdots,n.$ In the motivating example presented in Section~\ref{motivating_eg}, for each patient $i$, the binary response ($Y_i$) is the presence ($1$) or absence ($0$) of CKD. The covariate of interest ($X_i$) is the presence ($1$) or absence ($0$) of diabetes mellitus. The nuisance variables ($Z_i$) are all the other factors in the data such as age, blood pressure, anemia, coronary heart disease, and more. We assume that each of the binary response variables is independently generated from the following logistic regression model: 
\begin{equation}\label{logistic}
    P[Y_i=1 \mid X_i,Z_i, \theta, \beta] = \frac{\exp(X_i \theta + Z_i^T\beta)}{1+\exp(X_i \theta + Z_i^T\beta)},
\end{equation}
for $i= 1, \dots, n$, where $\theta$ and $\beta$ are random parameters having dimensions one and $d$, respectively. 

We consider the problem of providing valid inference for the low dimensional parameter $\theta$ in the above high dimensional logistic regression model where the dimension $d$ of the nuisance parameter $\beta$ can be much larger than the sample size.  Our goal is to construct a posterior distribution for the parameter $\theta$ given the data $(Y, X, Z)$ that can be used for obtaining interval estimators having valid frequentist coverage. We would like the validity of our inferential results to hold for all values of $\theta$. 

\subsection{Connections to Causal Inference}

The parameter of interest, $\theta$ in \eqref{logistic} has a casual interpretation as well. We can explicitly write down the expression for $\theta$ from Equation~\eqref{logistic} as:

\begin{align}\label{odds-ratio}
    \theta &= \log \bigg( \dfrac{P[Y_i = 1 \mid X_i = 1, Z_i]/P[Y_i = 0 \mid X_i = 1, Z_i]}{P[Y_i = 1 \mid X_i = 0, Z_i]/P[Y_i = 1 \mid X_i = 0, Z_i]}  \bigg) \nonumber \\
    &= \log \bigg( \dfrac{\text{odds}(Y_i \mid X_i = 1, Z_i)}{\text{odds}(Y_i \mid X_i = 0, Z_i)}  \bigg).
\end{align}
Based on \eqref{odds-ratio}, the parameter of interest $\theta$ can be viewed as the log-odds ratio in this setting where the covariate of interest $X_i$ is binary (or dichotomous). 

In the context of causal inference for observational studies \citep{cochran1965planning,rosenbaum2002overt}, we can refer to $X_i$ as the treatment variable for example an indicator for the administration of a certain drug or a medical procedure, and the outcomes $Y_i$ to be the binary health status of the patient. The high-dimensional nuisance covariates $Z_i$ can be treated as the confounding variables. Furthermore, consider the potential outcomes framework from the causal inference literature \citep{rubin1974estimating,imbens2015causal}. Under this framework, $Y_i(1)$ and $Y_i(0)$ represent the potential outcomes under treatment ($X_i =1$) and no treatment ($X_i = 0$), respectively.
Consider the following logistic models for the potential outcomes:
\begin{equation}\label{potential_outcome_model}
    \log \Bigg( \frac{P(Y_i(0) = 1 \mid Z_i)}{P(Y_i(0) = 0 \mid Z_i)} \Bigg) = Z_i^T\beta \text{ and } 
    \log \Bigg( \frac{P(Y_i(1) = 1 \mid Z_i)}{P(Y_i(1) = 0 \mid Z_i)} \Bigg) = \theta + Z_i^T\beta.
\end{equation}
% as assumed in \citep{zhu2019estimating}:
Based on the model assumption in \eqref{potential_outcome_model}, $\theta$ is a causal estimand often referred to as the conditional log-odds ratio in the casual inference literature pertaining to the observational studies \citep{causal_iv_clor,zhu2019estimating}. We can express $\theta$ as:
\begin{align}\label{causal_model_or}
    \theta &= \log \bigg( \dfrac{P[Y_i(1) = 1 \mid Z_i]/P[Y_i(1) = 0 \mid Z_i]}{P[Y_i(0) = 1 \mid Z_i]/P[Y_i(0) = 1 \mid Z_i]}  \bigg) \nonumber \\
    &= \log \bigg( \dfrac{\text{odds}(Y_i(1) \mid Z_i)}{\text{odds}(Y_i(0) \mid Z_i)}  \bigg).
\end{align}
Moreover, the standard assumptions of consistency (each subject is either treated or not) and strong ignorability (no unmeasured confounders) from the causal literature \citep{zhu2019estimating} imply that the model in \eqref{potential_outcome_model} can be expressed as a logistic model that resembles \eqref{logistic}:
\begin{equation}\label{causal_model}
    \log \Bigg( \frac{P(Y_i = 1 \mid X_i, Z_i)}{P(Y_i = 0 \mid X_i, Z_i)} \Bigg) = X_i  \theta + Z_i^T\beta.
\end{equation} 
Therefore, our goal of obtaining valid interval estimators for $\theta$ in \eqref{logistic} can lead to achieving a valid inference for the conditional log-odds ratio under the logistic model for the potential framework given in \eqref{potential_outcome_model}. This causal estimand can be helpful in deciding whether a particular treatment or drug is helpful for patients with particular characteristics \citep{conditional_odds_ratio}.

\subsection{Literature Review}

% \begin{enumerate}
%     \item High-dimensional literature review
%     \begin{itemize}
%         \item Methods from frequentist and Bayesian paradigm
%         \item Issues with such methods while performing inference -- larger bias, minimum signal strength condition
%     \end{itemize}
%     \item Methods for valid inference
%     \begin{itemize}
%         \item De-biasing methods, projection-based methods, selection methods
%         \item Many of these methods are tailored for continuous covariate of interest.
%     \end{itemize}
%     \item Focus on the parameter itself.
%     \begin{itemize}
%         \item Methods that focus on causal log-odds ratio
%         \item Have a couple of references, need to add more
%     \end{itemize}
% \end{enumerate}

In this paper, we consider the settings where the number of nuisance covariates $d$ can increase at a sub-exponential rate w.r.t. the sample size $n$. The model represented by Equation~\ref {logistic} is a high-dimensional logistic regression model that belongs to a bigger class of high-dimensional Generalized Linear Models (GLMs). Within the context of high-dimensional settings, the model parameters are subjected to sparsity assumptions to ensure the validity of estimation. In the frequentist paradigm, various methods have been proposed for GLMs that use an extra regularization penalty term to impose shrinkage and sparsity and therefore provide estimates and inference for the model parameters \citep{geerGLM, friedman2010regularization, breheny2011coordinate, huang2012estimation}. Approaches have been developed for logistic regression in particular \citep{bunea2008honest, bach2010self, kwemou2016non}. In the Bayesian paradigm, various methods have been proposed for high-dimensional estimation and selection using priors that induce shrinkage or sparsity \citep{george1993variable, israo,
park2008Bayesian, liang2008mixtures,
carvalho2009handling,  johnson2012Bayesian, bondell2012consistent,armagan2013generalized,     rovckova2014emvs, basad, bhattacharya2015dirichlet}.  Bayesian approaches have been developed for GLMs too and logistic regression in particular \citep{o2004bayesian,genkin2007large, ghosh2018use, skinny}.

Therefore, one can treat the parameter of interest $\theta$ as one of the parameters in a high-dimensional model and use the aforementioned methods to come up with an estimate and provide inferential results for $\theta$. However, in high-dimensional settings, the bias incurred by high-dimensional methods is $O(\sqrt{\log(d \vee n )/ n})$ which is much larger than the desirable rate of $O(n^{-1/2})$ for the low-dimensional settings. Moreover, these methods require the minimum signal strength of the nonzero coefficients to be sufficiently bounded away from zero to achieve theoretical consistency. While this assumption holds merit for variable selection purposes, it may not be suitable for inferential tasks. The impact of variable selection on inference has been well studied in the literature \citep{potscher1991effects,kabaila1995effect}. Several examples have been investigated where if the minimum signal strength condition is not met, the estimates are no longer $\sqrt{n}$-consistent, they fail to attain asymptotic normality, and therefore the inferential results breakdown \citep{leeb2005model, leeb2008sparse, potscher2009distribution,potscher2009confidence}.  

In the past decade, several de-biasing methods that do not rely on any minimum signal strength condition have been proposed. These methods aim to achieve $\sqrt{n}$-consistency and offer valid inference for parameters in high-dimensional models \citep{taylor2015statistical, panigrahi2018scalable, panigrahi2021integrative, van2014asymptotically, javanmard2014confidence, zhang2014confidence, belloni2010lasso, belloni2012sparse, berk2013valid, logistic_dml, belloni2014inference, dml, uposi, wang2020debiased}. Recently, a Bayesian method has been introduced, which possesses valid inference properties in the context of linear models and serves as the methodological motivation for the present paper \citep{qbayesian}.

Our specific aim is to develop a Bayesian methodology for inferring a low-dimensional parameter related to a categorical variable within a high-dimensional logistic model, a topic not previously explored in the literature. In the realm of linear models, \citet{qbayesian} devised a Bayesian approach grounded in conditional posteriors, demonstrating optimal convergence rates for the posterior distribution ($O(n^{-1/2}$) rate) and the validity of credible intervals. However, this method heavily depends on linearity and cannot be directly extended to logistic regression. In logistic regression, notable contributions have emerged from the frequentist paradigm. \citet{logistic_dml} utilized instrumental variables and Neyman orthogonality \citep{neyman1959optimal,neyman1979c} to create a double selection algorithm for inferring low-dimensional parameters in high-dimensional logistic models. A Bayesian method with valid inference in a high-dimensional logistic regression setting which combines the idea of Neyman orthogonality and conditional posteriors has been proposed \citep{ojha2023conditional}. However, these approaches rely significantly on the continuous nature of the covariate of interest, making them less efficient for the case when the covariate of interest is binary or categorical. In the realm of causal inference, Bayesian model averaging methods have been developed \citep{torrens2021confounder,antonelli2022causal}. Yet, these works focus on the average treatment effect, distinct from our parameter of interest $\theta$. A recent proposal by \citet{zhu2019estimating} aimed to estimate the conditional log-odds ratio \eqref{causal_model_or}, imposing stringent assumptions on nuisance covariates and their dependence structure with the covariate of interest. Additionally, several projection-based strategies within logistic regression have been introduced to de-bias estimates and construct valid confidence intervals \citep{shi_recursive_scores,ht,other}. However, all these recent approaches for inference in high-dimensional logistic regression operate within the frequentist framework. In contrast, our proposed approach strives to establish a Bayesian paradigm with valid frequentist properties.

\subsection{Our Contributions}

We employ Neyman's orthogonality principle \citep{neyman1959optimal,neyman1979c} to establish valid inference for the parameter of interest, denoted as $\theta$, which depends on the construction of an orthogonal score.
For logistic regression, \citep{logistic_dml} proposed a double selection algorithm that leverages the orthogonality principle to provide valid inference for low-dimensional parameters in high-dimensional settings. In their double selection algorithm, \citet{logistic_dml} perform a weighted regression of the covariate of interest, denoted as $X$, on the nuisance covariates, denoted as $Z$, to achieve orthogonality. The strategy discussed in \citet{ojha2023conditional}, based on the conditional posterior, is founded upon a weighted regression technique of similar nature. The moment conditions needed for efficient utilization of such a weighted regression and the validity of the theoretical results are contingent on the continuous nature of the covariate of interest, $X$.
However, when the covariate of interest is binary, it can be demonstrated that these moment conditions no longer hold. Consequently, the fundamental concept of weighted regression of $X$ on $Z$ cannot be utilized to achieve orthogonality in the case of binary $X$.

Our contributions are twofold: first, we derive a novel orthogonal score for the case of the binary covariate of interest, $X$. The construction of the orthogonal score is based on a variance-weighted projection. Then, we utilize this proposed score to formulate a Bayesian method that provides an estimate and constructs credible intervals for the parameter of interest, $\theta$. Moreover, we demonstrate that the constructed posterior for the parameter of interest, based on the proposed score, achieves asymptotic normality (uniformly at a $n^{-1/2}$ rate) with valid inference properties. The proposed Bayesian method can incorporate prior information about the parameter of interest.

In contrast to existing projection-based methods \citep{ht, other}, which rely on a single nuisance estimate for the de-biasing strategy, our proposed Bayesian method utilizes multiple samples of the nuisance parameters. Consequently, it captures multiple instances of orthogonality. This behavior resembles the averaging nature of integrated likelihood methods, which have been shown to quantify uncertainty in a superior way when removing nuisance parameters from a model \citep{berger1999integrated,severini2007integrated}. This intuition is supported by the results of our simulation studies and real data analysis, where the proposed method exhibits competitive performance compared to existing methods.

In Section~\ref{conditional_posterior}, we discuss the concept of conditional posteriors, develop a novel orthogonal score for $\theta$, and introduce the proposed Bayesian methodology.
In Section~\ref{theoretical_results}, we discuss the regularity assumptions and outline the theoretical properties of our estimator.
In Section~\ref{ordinal_extension}, we extend the proposed method to accommodate categorical covariates of interest, $X$.
In Section~\ref{simulation}, we conduct simulation studies to assess the performance of the proposed method, comparing it with existing Bayesian and frequentist methods. Additionally, we investigate the impact of mothers' smoking habits on infants' health indicators through semi-synthesized data.
In Section~\ref{conclusion}, we provide a final conclusion.
In Appendix~\ref{sampler}, we present the Gibbs sampler for posterior sampling.
In Appendix~\ref{proofs}, we provide proofs of all the theoretical results.

\section{Bayesian Conditional Posterior for Inference}\label{conditional_posterior}

Before we provide a detailed explanation of the proposed method, we discuss the idea of a ``conditional posterior'' as used in \citet{qbayesian}. In a standard Bayesian procedure, certain model assumptions are made that lead to a likelihood of the parameters. Given prior distributions and the likelihood for the parameters based on the observed data, we use Bayes' theorem to arrive at the posterior distribution for the model parameters.

The intuition behind the conditional posteriors, as used in the context of this paper, is slightly different from the standard Bayesian procedure outlined above. We start with a prior distribution for the model parameters and given data we directly propose a posterior distribution for the parameters. We refer to this posterior distribution as the ``conditional posterior'' that reflects our updated belief regarding the parameter based on the observed data and conditioning parameters. It is important to note that the data component of the conditional posterior may correspond to a specific working model. The proposal for the conditional posterior and the choice of the working model are driven by theoretical and practical considerations such as the attainment of valid inferential properties and the facilitation of efficient sampling. Our proposed method utilizes the Neyman orthogonality principle for the construction of our conditional posterior that yields valid inference for the parameters of interest.

\subsection{Notations}

Recall that $\theta$ is the parameter of interest and $\beta$ is the high-dimensional nuisance parameter. We use the notations $\theta_0$ and $\beta_0$ to denote the corresponding oracle quantities. We denote the diagonal covariance matrix of the binary output $Y$ (given $X,Z$) using $W$. The diagonal entries of $W$ are defined in terms of $(\theta_0,\beta_0)$ as:
\begin{equation}\label{w_true}
    W_{{i,i}} = \frac{\exp(X_i\theta_0 + Z_i^T\beta_0)}{(1+\exp(X_i\theta_0 + Z_i^T\beta_0))^2},
\end{equation}
for $i = 1,2,\dots,n$. 

We introduce functions of the nuisance covariates $\{h_0(Z_1), \dots, h_0(Z_n)\}$ as defined in \eqref{h_fn} to construct an orthogonal score \eqref{score_proposal} for the parameter of interest $\theta$. We refer to $h_0(Z_i)$ as the variance-weighted projections. The expression of $h_0(Z_i)$ requires knowledge of $P(Y_i \mid X_i, Z_i)$ and $P(X_i \mid Z_i)$ which we do not have given the observed data.  We refer to $P(X_i=1 \mid Z_i)$ as propensity scores for $i= 1, \dots, n$.

We use a sample of the parameters $(\theta, \beta)$ to obtain an estimator for $P(Y_i \mid X_i, Z_i)$. However, it is important to note that the sample of $\theta$ used in \eqref{p_y_def} is not the final sample used for inference on the parameter of interest. To emphasize this distinction, we denote this sample with an additional tilde symbol overhead, representing it as $\tilde{\theta}$. Therefore, the estimator for $P(Y_i \mid X_i, Z_i)$ is defined based on $(\tilde{\theta},\beta)$ as given in Equation~\eqref{p_y_def}. 
To obtain an estimator for the propensity scores $P(X_i=1 \mid Z_i)$, we introduce a high-dimensional parameter $\gamma$ which captures the dependence of the target covariate $X$ on the nuisance covariates $Z$ through a working logistic regression model.
Based on the estimators of $P(Y_i \mid X_i, Z_i)$ and $P(X_i=1 \mid Z_i)$, we obtain an estimator of $h_0(Z_i)$ which we denote by $h(Z_i)$ and is defined by Equation~\eqref{h_estimate}.

Based on $(\Tilde{\theta},\beta)$ and $h(Z_i)$, we define a re-parameterized nuisance quantity $\phi_i$  as $\phi_i = \Tilde{\theta}h(Z_i) + Z_i^T\beta$. Moreover, we denote the oracle values of $\phi_i$ using $\phi_{0i}$ which is defined as $\phi_{0i} = \theta_0h_0(Z_i) + Z_i^T\beta_0$. We use the notations $h_0(Z)$ and $h(Z)$ to represent the vectors $(h_0Z_1), \dots, h_0(Z_n))$ and $(h(Z_1), \dots, h(Z_n)),$ respectively. Furthermore, use the notations $\phi$ and $\phi_0$ to represent the vectors $(\phi_1,\dots, \phi_n)$ and $(\phi_{01},\dots, \phi_{0n}),$ respectively.
We use the symbol ``\textbf{WM}'' to denote the working models and the symbol ``\textbf{CP}'' to denote the conditional posteriors for the parameters.

\subsection{Formulation of Proposed Conditional Posteriors}

\subsubsection{\textbf{Issues with the existing strategy for orthogonality:}}

When the covariate of interest is continuous, the concept of orthogonality relies on a weighted regression between the covariate of interest and the nuisance covariates. However, when the covariate of interest is binary, such a weighted regression approach is no longer applicable. Firstly, regressing binary covariates on nuisance covariates may result in negative predictions for the binary covariate of interest, which are not interpretable. Moreover, to achieve orthogonality using the idea of weighted regression, the following assumption is crucial:
\begin{equation}\label{weighted_assumption}
    E[(X_i-Z_i^T\gamma_0)W_{i,i}Z_i \mid Z_i] = 0,
\end{equation}
where $\gamma_0$ is the coefficient of the weighted regression. When $X_i$ is binary, we can explicitly calculate the left-hand side of the equation in \eqref{weighted_assumption}. After expanding the expectation in \eqref{weighted_assumption}, we want the following expression to hold:
% \begin{align}\label{weighted_assumpn_expansion}
%     &E[\{Var(Y_i \mid X_i=1, Z_i)P(X_i =1 \mid Z_i)\}Z_i \mid Z_i] \\
%     &= E[\{Var(Y_i \mid X_i=1, Z_i)P(X_i =1 \mid Z_i) (Z_i^T\gamma_0)\} Z_i \mid Z_i] \nonumber\\
%     &+ E[\{Var(Y_i \mid X_i=0, Z_i)P(X_i =0 \mid Z_i)(Z_i^T\gamma_0)\} Z_i \mid Z_i] \nonumber
% \end{align}
\begin{align}\label{weighted_assumpn_expansion}
    Var(Y_i \mid X_i=1, Z_i)P(X_i =1 \mid Z_i)
    &= Var(Y_i \mid X_i=1, Z_i)P(X_i =1 \mid Z_i) (Z_i^T\gamma_0)\\
    &+ Var(Y_i \mid X_i=0, Z_i)P(X_i =0 \mid Z_i)(Z_i^T\gamma_0),\nonumber
\end{align}
where $Var(Y_i \mid X_i, Z_i) = P(Y_i=1 \mid X_i, Z_i)P(Y_i=0 \mid X_i, Z_i)$.
The requirement in \eqref{weighted_assumpn_expansion} can be simplified to the following requirement:
\begin{equation}\label{weighted_assumption_final_form}
    Z_i^T\gamma_0 = \frac{Var(Y_i \mid X_i=1, Z_i)P(X_i =1 \mid Z_i)}{Var(Y_i \mid X_i=1, Z_i)P(X_i =1 \mid Z_i) + Var(Y_i \mid X_i=0, Z_i)P(X_i =0 \mid Z_i)}.
\end{equation}
The left-hand side of \eqref{weighted_assumption_final_form} represents an unbounded linear function of the nuisance covariates, while the right-hand side of \eqref{weighted_assumption_final_form} represents a non-linear function of nuisance covariates bounded between zero and one. Therefore, fulfilling this requirement is not possible. As a result, we must explore an alternative strategy to construct an orthogonal score.

\subsubsection{\textbf{Construction of an orthogonal score using variance-weighted projection:}}

We first discuss the construction of an orthogonal score in this setting.
The usual score based on the model \eqref{logistic} can be expressed as:
\begin{equation}\label{score_og}
    \varphi_{usual}(\theta \mid X,Y,Z,\beta) = \sum_{i=1}^n \frac{\partial}{\partial \theta}\log P[Y_i =1 \mid X_i,Z_i, \theta, \beta] = X^T(Y-\mu),
\end{equation}
where $\mu_{n \times 1}$ is the logistic mean vector with $\mu_i = \exp(\theta X_i + Z_i^T\beta)/(1 +\exp(\theta X_i + Z_i^T\beta) ) $. Consider $W_{n \times n}$ a diagonal matrix as defined in Equation~\eqref{w_true}. Derivative of this score \eqref{score_og} with respect to the nuisance parameter $\beta$, when evaluated at $\theta_0$ and $\beta_0$ leads to:
\begin{equation}\label{usual_derivative}
    \partial_{\beta}[\varphi_{usual}]\big|_{\theta_0,\beta_0} = -X^TW Z, \text{ and therefore } E\big[\partial_{\beta}(\varphi_{usual})\big|_{\theta_0,\beta_0}\big]\neq 0.
\end{equation}
Therefore, the usual score does not have the orthogonal property. 

With the same definition of $\mu$ as in \eqref{score_og}, consider the following score for $\theta$:  
\begin{equation}\label{score_proposal}
    \psi(\theta \mid X,Y,Z,\beta) = (X- h_0(Z))^T(Y- \mu),
\end{equation}
where $h_0(Z):= (h_0(Z_1), h_0(Z_2), \dots, h_0(Z_n))$ is a non-parametric function of the nuisance covariates that has been introduced to impose orthogonality. We desire the following property to achieve the orthogonality for this proposed score:
\begin{equation}\label{orthogonality}
    E\big[\partial_{\beta}(\psi(\theta \mid X,Y,Z,\beta))\big|_{\theta_0,\beta_0}\big] =  0.
\end{equation}

The following expression for the function $h_0(Z_i)$ (for $i = 1,\dots, n$) leads to the desired orthogonal property:

\begin{align}\label{h_fn}
    h_0(Z_i) &= \frac{P(Y_i=1\mid X_i=1, Z_i)P(Y_i=0\mid X_i=1, Z_i)P(X_i=1\mid Z_i)}{\sum_{k \in \{0,1\}}P(Y_i=1\mid X_i=k, Z_i)P(Y_i=0\mid X_i=k, Z_i)P(X_i=k\mid Z_i)} \nonumber\\
    &= \frac{1}{1 + R_{0i}},\\
    R_{0i} &= \dfrac{P(Y_i=1\mid X_i=0, Z_i)P(Y_i=0\mid X_i=0, Z_i)P(X_i=0\mid Z_i)}{P(Y_i=1\mid X_i=1, Z_i)P(Y_i=0\mid X_i=1, Z_i)P(X_i=1\mid Z_i)},
\end{align}
where the quantities $P(Y_i \mid X_i = 1, Z_i)$ and $P(Y_i \mid X_i = 0, Z_i)$ are evaluated at the oracle values of the model parameters $(\theta_0,\beta_0)$. We have suppressed the use of extra notations in their respective conditional expression.
% Furthermore, the term $P(X_i =1 \mid Z_i)$ is often referred to as the propensity score in the causal literature. 
According to the expression in equation \eqref{h_fn}, it is evident that $h_0(Z_i)$ represents the conditional expectation of $X$, weighted by the variance of the output in the two classes. Consequently, we refer to the function $h_0(Z)$ as \textit{variance-weighted projection}.

Given the observed data $(Y,X,Z)$, we do not have the knowledge of $P(Y_i = 1 \mid X_i, Z_i)$ and $P(X_i = 1 \mid Z_i)$. These two quantities are sufficient to obtain an estimator of $h_0(Z_i)$. Therefore, we first need to obtain estimators for these quantities which will lead to an estimator of $h_0(Z_i)$. 

\subsubsection{Estimator of $P(Y_i = 1 \mid X_i, Z_i)$:}

We use a sample of the parameters $(\Tilde{\theta}, \beta)$ to obtain estimators for $P(Y_i = 1 \mid X_i, Z_i)$. The working model for $(\tilde{\theta},\beta)$ is the same as the original logistic model given by \eqref{logistic}. Consequently, we obtain samples of $(\tilde{\theta},\beta)$ based on the following conditional posterior:
\begin{equation}\label{beta_posterior}
\tag{$(\tilde{\theta}, \beta)$-\textbf{CP}}
\boxed{
    f((\tilde{\theta}, \beta) \mid Y,X, Z) \propto \pi(\tilde{\theta})\pi( \beta)\prod_{i=1}^{n}\frac{\exp(X_i\tilde{\theta} + Z_i^T\beta)^{
    Y_i}}{1 + \exp(X_i\tilde{\theta}+ Z_i^T\beta)},}
\end{equation}
where $\pi(\tilde{\theta})$ and $\pi( \beta)$ are the priors on the model parameters $\theta$ and $\beta$, respectively. The tilde notation in $\Tilde{\theta}$ emphasizes that this sample is different from the final sample used for inference on the parameter of interest defined later by Equation \eqref{theta_posterior}.

The estimators for $P(Y_i = 1 \mid X_i, Z_i)$ corresponding to $X_i=1$ and $X_i = 0$, respectively are defined as:
\begin{align}\label{p_y_def}
    P(Y_i = 1 \mid X_i = 1, Z_i, \tilde{\theta},\beta) &= \frac{\exp(\tilde{\theta} + Z_i^T\beta)}{1+\exp(\tilde{\theta} + Z_i^T\beta)} \nonumber \text{, and } \\
    P(Y_i = 1 \mid X_i = 0, Z_i,\tilde{\theta},\beta) &= \frac{\exp(Z_i^T\beta)}{1+\exp(Z_i^T\beta)}.
\end{align}

\subsubsection{Estimator for the propensity score, $P(X_i=1 \mid Z_i)$:}

We introduce another nuisance parameter $\gamma$ to model the dependence between $X$ and $Z$ and therefore, obtain an estimator for $P(X_i=1 \mid Z_i)$. The working model for parameter $\gamma$ is given by the following logistic model for $i = 1, \dots, n$:

\begin{equation}\label{gamma_working}
\tag{$\gamma$-\textbf{WM}}
    P(X_i = 1 \mid Z_i, \gamma) = \frac{\exp( Z_i^T\gamma)}{1+\exp(Z_i^T\gamma)}.
\end{equation}
Motivated by\eqref{gamma_working}, we sample  $\gamma$ based on the following conditional posterior:

\begin{equation}\label{gamma-posterior}
\tag{$\gamma$-\textbf{CP}}
\boxed{
    f(\gamma \mid X,Z) \propto \pi(\gamma)\prod_{i=1}^{n}\frac{\exp(Z_i^T\gamma)^{
    X_i}}{1 + \exp(Z_i^T\gamma)},}
\end{equation}
where $\pi(\gamma)$ is prior on the parameter $\gamma$.
Based on a sample of $\gamma$ from \eqref{gamma-posterior}, the estimators for the propensity scores are given by the expression of the working model \eqref{gamma_working}.

\subsubsection{Estimator of the variance-weighted projection:}

Based on a sample of $(\tilde{\theta}, \beta)$ from \eqref{beta_posterior}, we can obtain the estimators for $P(Y_i=1 \mid X_i=1, Z_i)$ and $P(Y_i=1 \mid X_i=0, Z_i)$ as defined in Equation~\eqref{p_y_def}. Moreover, based on a sample of $\gamma$ from \eqref{gamma-posterior}, we can obtain an estimator for $P(X_i =1 \mid Z_i)$ as defined through the working model \eqref{gamma_working}. For $i= 1, \dots, n$, we can obtain an estimator for $h_0(Z_i)$ which is denoted by $h(Z_i)$ and is defined as:

\begin{align}\label{h_estimate}
    h(Z_i) &= \frac{1}{1+ R_i}\\
    R_i &= \dfrac{P(Y_i=1\mid X_i=0, Z_i, \tilde{\theta},\beta)P(Y_i=0\mid X_i=0, Z_i, \tilde{\theta},\beta)P(X_i=0\mid Z_i, \gamma)}{P(Y_i=1\mid X_i=1, Z_i,\tilde{\theta},\beta)P(Y_i=0\mid X_i=1, Z_i,\tilde{\theta},\beta)P(X_i=1\mid Z_i,\gamma)}.\nonumber
\end{align}
In summary, a sample for the estimator of the variance-weighted projection is obtained from a sample of the model parameters $(\tilde{\theta}, \beta)$, which is drawn from \eqref{beta_posterior}, and a sample of $\gamma$ drawn from \eqref{gamma-posterior}.

\subsubsection{Working model for the parameter of interest $\theta$:}

We first define a re-parameterized nuisance quantity based on the samples $(\tilde{\theta},\beta)$ \eqref{beta_posterior} and the estimator $h(Z)$ \eqref{h_estimate} given by:

\begin{equation}\label{phi_def}
   \phi = (\phi_1, \dots, \phi_n), \text{ where } \phi_i = \tilde{\theta}h(Z_i) + Z_i^T\beta, \text{ for } i=1, \dots,n.
\end{equation}
Note that we use $\tilde{\theta} $ in the estimation of this quantity $\phi$ as it is sampled based on \eqref{beta_posterior} and is different from the conditional posterior for $\theta$ used for final inference. This makes the posterior concentration for $\phi$ and $\theta$ not depend on each other.
We assume the following working model for the parameter of interest $\theta$ conditional on $h(Z)$  and $\phi$:

\begin{equation}\label{theta_working}
\tag{$\theta$-\textbf{WM}}
    P(Y_i=1 \mid X_i,Z_i, h(Z_i),\phi_i, \theta) = \frac{\exp\{(X_i - h(Z_i))\theta + \phi_i\}}{1 + \exp\{(X_i - h(Z_i))\theta + \phi_i\} },
\end{equation}
for $i=1, \dots, n$. 
An important observation is that the working model  \eqref{theta_working} given the oracle quantities $\theta_0$, $h_0(Z)$, and $\phi_0$ is exactly same as the original model \eqref{logistic}. To see this:
\begin{align*}
    P(Y_i=1 \mid X_i,Z_i,\theta_0, h_0(Z), \phi_0) &= \frac{\exp\{ ( X_i - h_0(Z_i))\theta_0 + \phi_{0i} \}}{1 + \exp\{ ( X_i - h_0(Z_i))\theta_0 + \phi_{0i} \}} \\
    &= \frac{\exp\{X_i \theta_0  + \phi_{0i} - \theta_0h_0(Z_i) \}}{1 + \exp\{X_i \theta_0  + \phi_{0i} - \theta_0h_0(Z_i) \}} \\
    &= \frac{\exp(X_i \theta_0 + Z_i^T\beta_0)}{1+ \exp(X_i \theta_0 + Z_i^T\beta_0)}.
\end{align*}
This implies that the working model \eqref{theta_working} retains the same form as the true model under the oracle values of the nuisance quantities. This correspondence is crucial to capture the essential features of the original model and to yield valid inference.
Furthermore, the corresponding score of the working model \eqref{theta_posterior} can be shown to possess orthogonal property w.r.t. $h(Z)$ and $\phi$ based on the definitions of the variance-weighted projection \eqref{h_fn} and the original logistic model assumption \eqref{logistic}.

\subsubsection{Conditional Posterior for $\theta$:}

Based on the working model \eqref{theta_working} and the samples of $h(Z)$ and $\phi$, we obtain a sample for the parameter of interest $\theta$ based on the following conditional posterior:
\begin{equation}\label{theta_posterior}
\tag{$\theta$-\textbf{CP}}
    f(\theta \mid Y, X, Z, h(Z), \phi) \propto \pi(\theta) \prod_{i=1}^{n}\frac{\exp\{(X_i - h(Z_i))\theta + \phi_i\}^{Y_i}}{1 + \exp\{(X_i - h(Z_i))\theta + \phi_i\} }.
\end{equation}
where $\pi(\theta)$ is the prior distribution on the parameter of interest $\theta.$

In conclusion, $\gamma$ is sampled from its conditional posterior \eqref{gamma-posterior}, $h(Z)$ is sampled based on \eqref{gamma-posterior} and \eqref{beta_posterior} together, $\phi$ is sampled based on $h(Z)$ and \eqref{beta_posterior} together, and $\theta$ is sampled based on \eqref{theta_posterior}. These three conditional posteriors together provide a joint posterior for all the quantities involved:
\begin{align}\label{joint_posterior}
    f(\theta,\gamma, h(Z),\phi \mid Y,X,Z) &\propto  f(\theta \mid Y,X,Z, h(Z), \phi)  \times  f(\phi \mid Y,X,Z, h(Z)) \nonumber\\
    & \times f(h(Z) \mid X,Y,Z, \gamma) \times f(\gamma \mid X,Z).
\end{align} 
Here, $f(\gamma \mid X,Z)$ is given by \eqref{gamma-posterior} and $f(\theta \mid Y,X,Z, h(Z), \phi)$ is given by \eqref{theta_posterior}. The conditional distributions $f(\phi \mid Y, X, Z, h(Z))$ and $f(h(Z) \mid X, Y, Z, \gamma)$ are both derived from \eqref{beta_posterior}. This is due to the fact that $h(Z)$ is a function of $(\tilde{\theta}, \beta)$ given $\gamma$, and similarly, $\phi$ is a function of $(\tilde{\theta}, \beta)$ given $h(Z)$. It is worth noting that the joint posterior \eqref{joint_posterior} is not based on a single working model for all the parameters together but rather utilizes different working models for each parameter separately as described earlier.

\section{Theoretical Guarantees}
\label{theoretical_results}
 \subsection{Notations and Definitions}

We use the term $\logit(x)$ to represent the logistic link function \eqref{logistic}.  $L_n(\theta \mid h(Z), \phi)$ denotes the conditional log-likelihood function of the parameter of interest $\theta$ for any given $h(Z)$ and $\phi$. 
$L_n(\theta \mid h_0(Z), \phi_0)$ denotes the conditional log-likelihood function given the oracle quantities $h_0(Z),\phi_0$. 
For sequences $a_n$ and $b_n$, $a_n = O(b_n)$ means $a_n/b_n \leq M$ for some $M>0$ and $a_n = o(b_n)$ means $a_n/b_n \xrightarrow{n \to \infty} 0$. Convergence in probability is denoted by the symbol $\prob$. $o_P(1)$ stands for convergence in probability to zero. $O_P(1)$ stands for stochastically bounded.

We first explicitly write the log-likelihoods.
Based on the \eqref{theta_posterior}, the conditional log-likelihood is:
% \begin{subequations}
\begin{align}\label{log-l-notation}
         L_n(\theta \mid h(Z), \phi)
         &=  \sum_{i=1}^{n} Y_i \{\theta (X_i - h(Z_i)) + \phi_i\}  - \log\{1+\exp(\theta (X_i - h(Z_i)) + \phi_i)\}.
\end{align}
Similarly, under $\gamma_0$ and $\phi_0$, the conditional log-likelihood can be written as:
\begin{align}\label{log-l-notation-true}
            L_n(\theta \mid h_0(Z), \phi_0)
         &=  \sum_{i=1}^{n} Y_i \{\theta (X_i - h_0(Z_i)) + \phi_{0i}\}  - \log\{1+\exp(\theta (X_i - h_0(Z_i)) + \phi_{0i})\}.
\end{align}

\subsection{Assumptions/ Regularity Conditions}

\begin{enumerate}[label=A\arabic*]
\item \textit{On Dimension of nuisance parameters:}
    \label{ass1}
    $\log d = o(n)$ as $n \to \infty$.

\item \textit{On the regularity of design: }
    \label{ass2}
    The nuisance variables are bounded, that is, $$\max\{|Z_{ij}|,1 \leq i\leq n, 1\leq j \leq d\} \leq C, \text{ for some } 0 < C < \infty.$$ 
\item \textit{On sparsity of true nuisance parameter: } \label{ass3} Suppose the number of non-zero elements in $\beta_0$ is $s$, then
    \begin{equation}\label{sparsity}
       s^2\log(d \vee n) = o(\sqrt{n}). 
    \end{equation}

\item \textit{Concentration of the marginal posterior of $(\tilde{\theta},\beta)$: }\label{ass4} 
We use a sample of $(\tilde{\theta},\beta)$ from \eqref{beta_posterior} to obtain an estimator for $h_0(Z)$ \eqref{h_fn} and define the re-parametrized quantity $\phi$ \eqref{phi_def}.
We assume the following concentration rates for the conditional posterior of $(\tilde{\theta},\beta)$ : 
\begin{align}
    P\bigg( \max \Big\{ \frac{\|(\tilde{\theta},\beta) - (\theta_0,\beta_0)\|_1}{s} ,\frac{\|(\tilde{\theta},\beta) - (\theta_0,\beta_0)\|_2}{s^{1/2}} \Big\} &\geq M \sqrt{\frac{\log(d \vee n)}{n}} \mid Y,X,Z \bigg) \nonumber\\
    &\leq C (d \vee n)^{-c_1},
\end{align}
    on a set $\mathcal{E}_1$, with $P(\mathcal{E}_1) \geq 1 - (d \vee n)^{-c_1}$, where $M,C,$ and $c_1$ are some positive constants.

\item \textit{Concentration of the estimator for the propensity scores} $P(X_i=1 \mid Z_i)$: \label{ass_propensity}
We use $P(X_i=1 \mid Z_i, \gamma)$ defined in \eqref{gamma_working} based on a sample of $\gamma$ from \eqref{gamma-posterior} as an estimator for the propensity scores $P(X_i = 1 \mid Z_i)$. We assume the following concentration rate:
\begin{align}
P\bigg( \max_{i \leq n} \Bigg| \log \Bigg(\frac{P(X_i=1 \mid Z_i, \gamma)/P(X_i=0 \mid Z_i, \gamma)}{P(X_i=1 \mid Z_i)/P(X_i=0 \mid Z_i)} \Bigg) \Bigg| &\geq M \sqrt{\frac{\log(d \vee n)}{n}} \mid X,Z \bigg) \nonumber\\
    &\leq C (d \vee n)^{-c_2},
\end{align}
    on a set $\mathcal{E}_2$, with $P(\mathcal{E}_2) \geq 1 - (d \vee n)^{-c_2}$, where $M,C,$ and $c_2$ are some positive constants.

\item \textit{On prior of the parameter of interest $\pi(\theta)$: } \label{ass5}
The prior density of the parameter of interest $\pi(\theta)$ is continuous at $\theta = \theta_0$, and $\pi(\theta_0) >0$. Furthermore, we assume that $|\theta_0| \leq M_0$ for some $M_0>0$.

\end{enumerate}

 Assumption~\ref{ass1} is the growth condition for the number of high dimensional covariates in the model. We consider the setting where the total number of covariates can increase as the sample size increases at a sub-exponential rate. 
 Assumption~\ref{ass2} requires that each component of the nuisance covariate vector is bounded. 
 
 Assumption~\ref{ass3} presents a bound on the level of sparsity for the nuisance parameters. When additional assumptions regarding the boundedness of the maximum eigenvalues of 
$Z$ are made, the sparsity assumption commonly used in the high-dimensional literature for estimation alone (and not for valid inference) is:  $s = o(n/ \log(d \vee n))$ \citep{candes2007dantzig,cai2017confidence}.
To achieve $\sqrt{n}$-consistency and valid inference properties, the sparsity assumption used in the literature is $s = o(\sqrt{n}/\log(d \vee n))$ \citep{van2014asymptotically,zhang2014confidence,logistic_dml,qbayesian}. If we further assume the boundedness of the maximum eigenvalue of $Z$, we can relax our sparsity assumption to the same order of $s = o(\sqrt{n}/\log(d \vee n))$ as well.
 
 Assumption~\ref{ass4} intuitively means that with high probability, the posterior of the nuisance parameter concentrates around their oracle values at $\sqrt{\log(d \vee n) /n}$ rate.
 This is the standard concentration rate in high-dimensional settings and is commonly satisfied by most regularization and selection methods both from the frequentist or Bayesian paradigms \citep{israo,scad,geerGLM,liang2008mixtures,castillo2015Bayesian,song2022nearly}. In this paper, we use the Spike and Slab method \citep{skinny} for the nuisance estimation which satisfies this concentration rate as well. However, it is important to note that this concentration rate does not guarantee valid inference, and the resulting estimates may exhibit significant bias.
 As mentioned in Section 1, this concentration rate is slower than the desired $n^{-1/2}$ rate of convergence.
 
 Assumption~\ref{ass_propensity} intuitively means that the ratio of our estimate for the propensity score and the corresponding oracle values concentrates around one at $\sqrt{\log(d \vee n) /n}$ rate. Consider a scenario in which the dependence between $X$ and $Z$ were given by a logistic model with parameter $\gamma_0$:
 $P(X_i=1\mid Z_i)= \exp(Z_i^T\gamma_0)/(1 + \exp(Z_i^T\gamma_0))$.
 Then this assumption is equivalent to the concentration of $\gamma$ around $\gamma_0$ at $\sqrt{\log(d \vee n) /n}$ rate. We do not want to assume any correct model for $P(X_i=1\mid Z_i)$. Therefore, we make an assumption about closeness between $P(X_i=1\mid Z_i)$ and our estimator $P(X_i=1\mid Z_i,\gamma)$.

 Assumption~\ref{ass5} guarantees the positive support of prior at $\theta_0$ which can be verified in the case of a Gaussian prior for example. We further assume that the oracle parameter of interest $\theta_0$ has bounded magnitude. 

\subsection{Main Results}
 When the true values of the nuisance quantities $h_0$ and $\phi_0$ are known, the conditional posterior of $\theta$ is given by:
    \begin{equation}\label{conditional_model_true}
    f(\theta \mid Y,X,Z, h_0(Z), \phi_0) \propto \pi(\theta) \prod_{i=1}^{n}\frac{\exp\{(X_i - h_0(Z_i))\theta + \phi_{0i}\}^{Y_i}}{1 + \exp\{(X_i - h_0(Z_i))\theta + \phi_{0i}\} }.
\end{equation}
This posterior is for one dimensional parameter $\theta$ and there is no variable selection involved. This corresponds to a logistic model with new input and a random effect term. Moreover, the log-likelihood, in this case, is given by \eqref{log-l-notation-true}:
\begin{align}\label{loglik}
            L_n(\theta \mid h_0(Z), \phi_0)
         &=  \sum_{i=1}^{n} Y_i \{(X_i - h_0(Z_i))\theta + \phi_{0i}\}  - \log\{1+\exp((X_i - h_0(Z_i))\theta + \phi_{0i})\}.
\end{align}
From a frequentist perspective, the corresponding score has an expectation zero at the oracle value $\theta_0$. Therefore, the posterior concentration theorem \citep{walker,schervish2012theory} applies here and the posterior concentrates around $\theta_0$ at the desired $n^{-1/2}$ rate. We denote the Maximum Likelihood Estimator for the likelihood in \eqref{loglik} by $\hat{\theta}_0$.
\begin{equation}\label{true_mle}
    \hat{\theta}_0 = \argmax_{\theta \in \Theta} \frac{1}{n}L_n(\theta \mid h_0(Z), \phi_0).
\end{equation}
This is the Maximum Likelihood Estimate of a one-dimensional parameter in a logistic model which does not have a closed-form expression. However, based on the theory of the Maximum Likelihood Estimates \citep{fahrmeir1985consistency}, we know that it concentrates around $\theta_0$ at the desired $n^{-1/2}$ rate. The standard deviation of $\hat{\theta}_0$ will be denoted by $\sigma_n$ and is defined as:
\begin{equation}\label{true_mle_sd}
\sigma_n = \sqrt{- \frac{\partial^2 }{\partial \theta^2}L_n(\theta \mid h_0(Z), \phi_0)\bigg|_{\theta = \hat{\theta}_0}}.
\end{equation}
We now provide the main theoretical results. The conditional posterior for $\theta$ as defined in \eqref{theta_posterior} is:
\begin{equation}\label{posterior}
    f(\theta \mid Y, X, Z, h(Z), \phi) = \frac{\pi(\theta)\exp\{L_n(\theta \mid h(Z), \phi)\}}{\int_{\Theta}\pi(\theta)\exp(L_n(\theta \mid h(Z), \phi)) d\theta},
\end{equation}
where $L_n$ is the log likelihood function based on the \eqref{theta_posterior} as defined in \eqref{log-l-notation}. Let $f(\theta \mid Y,X,Z)$ be the marginal posterior for $\theta$ given data. Our theorem below will guarantee the concentration of this posterior for $\theta$ at the optimal rate.

\begin{theorem}[\textbf{Posterior Concentration}]\label{theorem1}

Suppose Assumptions~\ref{ass1}-\ref{ass5} presented above hold. Consider $\hat{\theta}_0$ and $\sigma_n$ as defined in \eqref{true_mle} and \eqref{true_mle_sd}, respectively. 
If $a$ and $b$ are constants, where $a<b$, then the marginal posterior probability $P(\mle + a\sigma_n < \theta < \mle + b\sigma_n \mid Y, X, Z)$ given by 
$$\int_{\hat{\theta}_0 + a\sigma_n}^{\hat{\theta}_0 + b\sigma_n} f(\theta \mid Y,X,Z)d\theta$$
converges in probability to 
$$\frac{1}{\sqrt{2\pi}} \int_{a}^{b} \exp\bigg(-\frac{1}{2}z^2 \bigg) dz,$$
as $n\to \infty$.
\end{theorem}

The theorem states that the posterior  concentrates around $\hat{\theta}_0$ which is the oracle estimator \eqref{true_mle} at $n^{-1/2}$ rate because $\sigma_n = O(\sqrt{n})$. 
The concentration of $\hat{\theta}_0$ around $\theta_0$ implies that the conditional posterior concentrates around the true parameter $\theta_0$ at $n^{-1/2}$ rate. As the conditional posterior attains asymptotic normality, the construction of confidence intervals would be valid as confirmed by the following corollary:

\begin{corollary}
  The $(1-\alpha)$-credible interval has the correct coverage in frequentist sense. Let $\hat{q}_{\alpha/2}^{CB}$ and $\hat{q}_{1 -\alpha/2}^{CB}$ be the $\alpha/2^{th}$ and $(1 -\alpha/2)^{th}$ percentiles of the marginal posterior distribution $f(\theta \mid Y,X,Z)$, respectively. Then, we have
 \begin{equation*}
     \Bigg| P \Big[\theta_0 \in \big(\hat{q}_{\alpha/2}^{CB},\hat{q}_{1 -\alpha/2}^{CB}\big) \mid Y, X, Z\Big] - (1-\alpha) \Bigg| \prob 0.
 \end{equation*}
\end{corollary}

The proofs of the theorem and the corollary have been presented in Section \ref{proofs}.

\section{Prior Specification}

We now provide our prior specifications which we use for our empirical investigations. 
\begin{enumerate}[label=P\arabic*]
    \item 
{\bf Prior for $\beta$}: Let $I_{1j}$ denote the binary latent variables that decide the active and inactive state of the component $\beta_j$. The priors are:
\begin{subequations}\label{eta-prior}
    \begin{equation}\nonumber
            \beta_j\mid I_{1j} = 0 \sim N(0,\tau_{0n}^2), \beta_j\mid I_{1j} = 1 \sim N(0,\tau_{1n}^2),
    \end{equation}
    \begin{equation}\nonumber
            P(I_{1j} = 1) = 1 - P(I_{1j} = 0) = q_{n},
    \end{equation}
\end{subequations}
    for $j = 1,\dots, d$. The parameters  $\tau_{0n},\tau_{1n}$, and $q_n$ are chosen exactly the same way as in \citep{skinny}. To be specific, we have chosen $\tau_{0n}^2 = 1/n$, $n\tau_{1n}^2 = \max\{n,0.01 d^{2.1} \}$, and   $q_{n}$ such that $P[\sum_{j=1}^d I_{1j} > K] = 0.1$ for $K= \max\{10,\log n\}$.

 \item {\bf Prior for $\gamma$:} Let $I_{2j}$ denote the binary latent variable that decides the active and inactive state of the component $\gamma_j$. The priors are:
    \begin{subequations}\label{gamma-prior}
    \begin{equation}\nonumber
            \gamma_j\mid I_{2j} = 0 \sim N(0,\tau_{0n}^2), \gamma_j\mid I_{2j} = 1 \sim N(0,\tau_{1n}^2),
    \end{equation}
    \begin{equation}\nonumber
            P(I_{2j} = 1) = 1 - P(I_{2j} = 0) = q_n,
    \end{equation}
    \end{subequations}
    for $j = 1,\dots, d$. The choice of the prior hyper-parameters is the same as that discussed in the prior for $\beta$.

 \item {\bf Prior for $\theta$:} This is a finite-dimensional parameter that does not require any variable selection step. We assume a Gaussian prior for this parameter with a large variance $\lambda >0:$ $\theta \sim N(0,\lambda).$ We have chosen $\lambda = 10$ in our simulations.
 %   \end{equation*}
 %   where the variance $\lambda$ is a large value.
\end{enumerate}

\section{Simulation Studies}\label{simulation}

\subsection{Simulation Setup}\label{methods}

We investigate the performance of our proposed methods under different simulation settings. Through simulation results, we see that the conditional Bayesian method achieves valid coverage in the frequentist sense. 

Our simulations are based on the following models:
\[ P[Y_i=1 \mid X_i,\textbf{Z}_i] = \logit(\theta_0 X_i + \beta^T \textbf{Z}_i), \ \ \ P[X_i=1 \mid Z_i] = \logit(\gamma^T \textbf{Z}_i),\]
where $\mathbf{Z_i}$ is the high-dimensional nuisance variable generated according to a multivariate normal, $N(0, H_{d \times d})$ with $H_{ij} = 0.5^{|i-j|}$. The nuisance parameters $\beta$ and $\gamma$ are chosen to be sparse vectors according to
\[ \beta_{d \times 1} = (-0.4,0.8,-1,1.5, 0, 0, \cdots , 0)^T,\gamma_{d \times 1} = (0.3,-0.5,-1,1.5,0, 0, \cdots, 0)^T. \]

%{\color{red} Add the prior parameter values.}\\

We consider two pairs of sample sizes and number of covariates such that $(n,d) \in \{(400,500),(500,600)\}$ and select the true signal $\theta_0$ from a set of varying signal strengths, $\theta_0 \in \{0, 0.1, 0.2,\dots, 1\}$. We calculate the empirical $95 \%$ credible (confidence) intervals for different values of $\theta_0$ based on 1000 Monte Carlo simulations. We report the frequentist coverage, interval length, and bias for each method under consideration. The methods that we compare in this study are:
\color{black}

\begin{itemize}
% \item \textbf{DS}: \textit{Double Selection} - To implement the Double Selection algorithm proposed by \citep{logistic_dml}, we used the pseudo-code given in Table 2 of their paper using the exact values of the penalty parameters they described.  
 \item \textbf{WLP}: \textit{Global and Simultaneous Hypothesis Testing for High-Dimensional Logistic Regression Models} \citep{ht} - This method uses a generalized low-dimensional projection technique for the bias correction of the parameter estimates which are obtained after employing logistic LASSO. They construct global test statistics and construct confidence intervals for all the parameters using the debiased estimates. We use the intervals and estimates for the parameter of interest $\theta$ in our results. This method has been implemented based on the code provided by the authors.
 
 \item \textbf{LSW}: \textit{Statistical Inference for High-Dimensional Generalized Linear Models With Binary Outcomes} \citep{other} - This is a two-step bias correction method that employs a novel weighting strategy for the bias correction and consequently construct confidence intervals. This method has been implemented based on the code provided by the authors. %This method has been proposed for several generalized linear models, however, we use the algorithm related to logistic regression in this paper.
 \item \textbf{NAIVE}: \textit{Post LASSO Logistic Regression} - In the first step of this method, we perform LASSO using the \texttt{glmnet} package to select the significant nuisance covariates. In the second step, we use the covariate of interest $X$ along with the selected nuisance covariates from the first step to perform a low-dimensional logistic regression and use the estimated parameter and its variance for inference using Wald intervals.
 
 \item \textbf{BMA}: \textit{Bayesian Model Averaging} \citep{bma_tuchler}:  This method utilizes  spike and slab priors over the high-dimensional nuisance parameter for model selection and subsequently averages over selected models to obtain estimates and standard deviation for the parameter of interest. The point estimator and the estimated standard deviation are used to obtain intervals. We implement this method using the package \texttt{BoomSpikeSlab}.
 
 \item \textbf{BLASSO}: \textit{Bayesian LASSO} \citep{park2008Bayesian} - This is the Bayesian logistic regression method with Laplace priors on the model parameters. We use the \texttt{bayesreg} package to implement this method.

 \item \textbf{CB}: This is our proposed Bayesian approach based on conditional posteriors using the \textit{Gibbs sampler}. 
 \item \textbf{ORACLE} - We regress the output $Y$ on the target covariate $X$ and the nuisance covariates $Z$ which truly affect $Y$. The selection of nuisance covariates is based on the indices where the components of $\beta$ are non-zero. As this corresponds to a low-dimensional logistic regression scenario, we employ the \texttt{glm} package without any form of regularization. We include this approach as a benchmark but it cannot be implemented in practice as it uses the knowledge of the unknown sparsity structure of the true data generating model.
\end{itemize}
%{\color{red} There is no numbering required for this above list of methods being compared. Also, maintain the same sequence as used in the results table.}

\color{black}

\subsection{Results and Discussion}

\begin{table}[ht]
\centering
\resizebox{0.9\linewidth}{!}{%
\begin{tabular}{|c|c|c|c|c|c|c|c|c|c|}
  \hline
  \multicolumn{2}{|c|}{}& \multicolumn{4}{|c|}{Frequentist}& \multicolumn{3}{|c|}{Bayesian} & \multicolumn{1}{|c|}{}\\
  \hline
Qunatities & $\theta_0$ & DS & NAIVE & WLP & LSW & BMA & BLASSO & CB & ORACLE   \\
  \hline
 & 0.0 & 0.507 & 0.909 & 0.871 & 0.979 & 1.000 & 0.994 & 0.950 & 0.943 \\ 
   & 0.1 & 0.521 & 0.906 & 0.927 & 0.979 & 0.001 & 0.998 & 0.953 & 0.953 \\ 
   & 0.2 & 0.553 & 0.889 & 0.935 & 0.977 & 0.003 & 0.998 & 0.961 & 0.942 \\ 
   & 0.3 & 0.551 & 0.887 & 0.965 & 0.971 & 0.002 & 0.997 & 0.971 & 0.950 \\ 
  Coverage & 0.4 & 0.573 & 0.861 & 0.974 & 0.964 & 0.011 & 0.992 & 0.949 & 0.950 \\ 
   & 0.5 & 0.565 & 0.859 & 0.973 & 0.965 & 0.022 & 0.980 & 0.965 & 0.964 \\ 
   & 0.6 & 0.523 & 0.835 & 0.959 & 0.956 & 0.035 & 0.968 & 0.953 & 0.947 \\ 
   & 0.7 & 0.535 & 0.823 & 0.939 & 0.947 & 0.114 & 0.965 & 0.964 & 0.958 \\ 
   & 0.8 & 0.529 & 0.800 & 0.910 & 0.940 & 0.186 & 0.970 & 0.956 & 0.961 \\ 
   & 0.9 & 0.520 & 0.771 & 0.888 & 0.938 & 0.311 & 0.956 & 0.947 & 0.939 \\ 
   \hline
\hline
 & 0.0 & 0.706 & 1.662 & 0.646 & 1.253 & 0.001 & 1.076 & 1.303 & 1.161 \\ 
   & 0.1 & 0.709 & 1.642 & 0.646 & 1.254 & 0.001 & 1.144 & 1.293 & 1.155 \\ 
   & 0.2 & 0.710 & 1.657 & 0.650 & 1.267 & 0.007 & 1.279 & 1.301 & 1.159 \\ 
  Interval & 0.3 & 0.714 & 1.635 & 0.653 & 1.275 & 0.009 & 1.363 & 1.300 & 1.158 \\ 
  Length & 0.4 & 0.718 & 1.674 & 0.660 & 1.266 & 0.030 & 1.490 & 1.300 & 1.160 \\ 
   & 0.5 & 0.721 & 1.728 & 0.661 & 1.260 & 0.075 & 1.605 & 1.302 & 1.162 \\ 
   & 0.6 & 0.727 & 1.847 & 0.677 & 1.265 & 0.126 & 1.767 & 1.317 & 1.169 \\ 
   & 0.7 & 0.732 & 1.766 & 0.677 & 1.278 & 0.314 & 1.905 & 1.320 & 1.172 \\ 
   & 0.8 & 0.738 & 1.798 & 0.698 & 1.272 & 0.501 & 2.090 & 1.326 & 1.180 \\ 
   & 0.9 & 0.745 & 1.857 & 0.702 & 1.269 & 0.695 & 2.215 & 1.330 & 1.185 \\ 
   \hline
\end{tabular}}
\caption{Coverage and interval length corresponding to sample-based standard errors for each method considered for $n=400, d=500$ under signal strengths $\theta_0 \in \{0,0.1,0.2,\dots,0.9\}$.}
\label{n-400}
\end{table}

\begin{table}[ht]
\centering
\resizebox{0.9\linewidth}{!}{%
\begin{tabular}{|c|c|c|c|c|c|c|c|c|c|}
  \hline
  \multicolumn{2}{|c|}{}& \multicolumn{4}{|c|}{Frequentist}& \multicolumn{3}{|c|}{Bayesian} & \multicolumn{1}{|c|}{}\\
  \hline
Qunatities & $\theta_0$ & DS & NAIVE & WLP & LSW & BMA & BLASSO & CB & ORACLE   \\
  \hline
 & 0.0 & 0.556 & 0.904 & 0.858 & 0.982 & 1.000 & 0.996 & 0.949 & 0.936 \\ 
   & 0.1 & 0.583 & 0.902 & 0.902 & 0.982 & 0.000 & 0.996 & 0.933 & 0.947 \\ 
   & 0.2 & 0.557 & 0.872 & 0.949 & 0.978 & 0.000 & 0.998 & 0.951 & 0.946 \\ 
   & 0.3 & 0.559 & 0.875 & 0.955 & 0.971 & 0.007 & 1.000 & 0.955 & 0.944 \\ 
  Coverage & 0.4 & 0.517 & 0.871 & 0.967 & 0.973 & 0.016 & 0.993 & 0.951 & 0.956 \\ 
   & 0.5 & 0.512 & 0.827 & 0.960 & 0.953 & 0.027 & 0.958 & 0.956 & 0.944 \\ 
   & 0.6 & 0.482 & 0.840 & 0.933 & 0.976 & 0.082 & 0.964 & 0.958 & 0.960 \\ 
   & 0.7 & 0.458 & 0.816 & 0.929 & 0.967 & 0.202 & 0.944 & 0.962 & 0.951 \\ 
   & 0.8 & 0.484 & 0.757 & 0.845 & 0.969 & 0.295 & 0.940 & 0.945 & 0.924 \\ 
   & 0.9 & 0.465 & 0.716 & 0.869 & 0.957 & 0.494 & 0.940 & 0.950 & 0.940 \\ 
   \hline
\hline
 & 0.0 & 0.640 & 1.414 & 0.584 & 1.515 & 0.001 & 0.908 & 1.150 & 1.036 \\ 
   & 0.1 & 0.644 & 1.349 & 0.585 & 1.517 & 0.001 & 0.992 & 1.144 & 1.030 \\ 
   & 0.2 & 0.644 & 1.459 & 0.589 & 1.520 & 0.002 & 1.096 & 1.148 & 1.030 \\ 
  Interval & 0.3 & 0.648 & 1.401 & 0.594 & 1.510 & 0.016 & 1.205 & 1.152 & 1.032 \\ 
  Length & 0.4 & 0.650 & 1.440 & 0.595 & 1.507 & 0.040 & 1.329 & 1.146 & 1.033 \\ 
   & 0.5 & 0.655 & 1.428 & 0.601 & 1.484 & 0.082 & 1.400 & 1.160 & 1.036 \\ 
   & 0.6 & 0.659 & 1.449 & 0.607 & 1.563 & 0.208 & 1.540 & 1.161 & 1.040 \\ 
   & 0.7 & 0.665 & 1.476 & 0.606 & 1.513 & 0.432 & 1.650 & 1.167 & 1.043 \\ 
   & 0.8 & 0.670 & 1.533 & 0.617 & 1.565 & 0.606 & 1.745 & 1.172 & 1.051 \\ 
   & 0.9 & 0.675 & 1.544 & 0.626 & 1.569 & 0.807 & 1.903 & 1.187 & 1.055 \\ 
   \hline
\end{tabular}}
\caption{Coverage and interval length corresponding to sample-based standard errors for each method considered for $n=500, d=600$ under signal strengths $\theta_0 \in \{0,0.1,0.2,\dots,0.9\}$.}
\label{n-500}
\end{table}

\begin{table}[ht]
\centering
\resizebox{0.9\linewidth}{!}{%
\begin{tabular}{|c|c|c|c|c|c|c|c|c|c|}
  \hline
  \multicolumn{2}{|c|}{}& \multicolumn{4}{|c|}{Frequentist}& \multicolumn{3}{|c|}{Bayesian} & \multicolumn{1}{|c|}{}\\
  \hline
$(n,d)$ & $\theta_0$ & DS & NAIVE & WLP & LSW & BMA & BLASSO & CB & ORACLE   \\
  \hline
 & 0.0 & 0.124 & 0.014 & 0.163 & 0.028 & 0.000 & 0.196 & 0.084 & 0.011 \\ 
   & 0.1 & 0.072 & 0.021 & 0.128 & 0.016 & 0.100 & 0.148 & 0.078 & 0.002 \\ 
   & 0.2 & 0.022 & 0.102 & 0.097 & 0.016 & 0.199 & 0.134 & 0.064 & 0.015 \\ 
   & 0.3 & 0.001 & 0.113 & 0.044 & 0.069 & 0.300 & 0.091 & 0.084 & 0.003 \\ 
  n=400 & 0.4 & 0.015 & 0.165 & 0.014 & 0.084 & 0.398 & 0.081 & 0.098 & 0.000 \\ 
  d=500 & 0.5 & 0.060 & 0.281 & 0.012 & 0.078 & 0.493 & 0.098 & 0.071 & 0.029 \\ 
   & 0.6 & 0.057 & 0.355 & 0.056 & 0.121 & 0.581 & 0.105 & 0.115 & 0.004 \\ 
   & 0.7 & 0.098 & 0.359 & 0.078 & 0.124 & 0.627 & 0.132 & 0.112 & 0.003 \\ 
   & 0.8 & 0.116 & 0.426 & 0.091 & 0.144 & 0.650 & 0.196 & 0.124 & 0.011 \\ 
   & 0.9 & 0.100 & 0.536 & 0.118 & 0.145 & 0.628 & 0.254 & 0.134 & 0.020 \\ 
   \hline
\hline
 & 0.0 & 0.036 & 0.005 & 0.142 & 0.009 & 0.000 & 0.164 & 0.074 & 0.008 \\ 
   & 0.1 & 0.053 & 0.033 & 0.104 & 0.018 & 0.100 & 0.118 & 0.074 & 0.006 \\ 
   & 0.2 & 0.100 & 0.111 & 0.060 & 0.011 & 0.200 & 0.093 & 0.067 & 0.003 \\ 
  n=500 & 0.3 & 0.120 & 0.162 & 0.042 & 0.057 & 0.299 & 0.071 & 0.068 & 0.010 \\ 
  d=600 & 0.4 & 0.165 & 0.116 & 0.005 & 0.076 & 0.397 & 0.067 & 0.069 & 0.007 \\ 
   & 0.5 & 0.192 & 0.201 & 0.042 & 0.093 & 0.493 & 0.042 & 0.091 & 0.015 \\ 
   & 0.6 & 0.169 & 0.280 & 0.076 & 0.112 & 0.562 & 0.064 & 0.090 & 0.001 \\ 
   & 0.7 & 0.200 & 0.322 & 0.082 & 0.097 & 0.571 & 0.106 & 0.093 & 0.005 \\ 
   & 0.8 & 0.206 & 0.402 & 0.117 & 0.136 & 0.556 & 0.121 & 0.116 & 0.009 \\ 
   & 0.9 & 0.229 & 0.463 & 0.117 & 0.126 & 0.479 & 0.231 & 0.108 & 0.019 \\ 
   \hline
\end{tabular}}
\caption{Bias for different methods under signal strengths $\theta_0 \in \{0,0.1,0.2,\dots,0.9\}$.}
\label{bias}
\end{table}

In this section, we compare the results from the simulation study for all the methods considered. Table~\ref{n-400} and Table~\ref{n-500} contain the coverage and interval lengths of methods under comparison for $(n=400,d=500
)$ and $(n=500,d=600)$ cases, respectively. Table~\ref{bias} contains the bias of the estimates produced by the methods under comparison based on 1000 Monte Carlo simulations.

Based on Tables~\ref{n-400}-\ref{bias}, we observe that the coverage of the proposed method (CB) is close to the desired coverage. The length of the intervals produced by the proposed method is slightly larger than the optimal ones obtained by ORACLE. While the bias attained by the proposed method is slightly larger than the ones obtained by ORACLE, it is smaller compared to the other methods for the most signal values. The coverage of BLASSO is at least as much as the nominal coverage, however, the bias and the length of the interval produced is much larger, especially for larger signal values. On the other hand, NAIVE and BMA suffer from under-coverage and large bias for most signal values. WLP attains smaller bias but suffers from under-coverage for most signal values possibly because the length of the intervals produced is much smaller compared to the optimal ones by ORACLE. Among the frequentist methods, LSW attains better coverage with the interval length and bias similar to the proposed method. However, for most signal values LSW has over-coverage and in the case of larger number of covariates (Table~\ref{n-500}), it yields intervals with much larger lengths compared to CB and ORACLE. 
In summary, the proposed method demonstrates
a competitive overall performance in terms of achieving the desired coverage, smaller
interval length and bias.

\subsection{Synthetic Data Analysis}
% w{\color{red} clean this section! }

% In this section, we investigate the empirical performance of the proposed method on real data. We consider the Natality birthweight dataset. Our goal is to provide valid inference for the impact of mothers' smoking habits on the mortality of infants. However, we do not have knowledge of the true impact of the mothers' smoking habit on the mortality of infants. To facilitate the calculation and comparison of coverage and bias, we simulate binary responses based on the covariates from the data. We use the results from studies about the impact of several maternal factors on infant mortality rate \citep{sabaphillip} and generate an indicator output. This output is not exactly the mortality rate, instead can be interpreted as
% a health indicator. In such a synthesized output scenario, we can control the impact of the mother's smoking habit on the output. We take the mother's smoking covariate from the real data. The working model \eqref{gamma_working} between the target (mothers' smoking in this case) and nuisance covariates (maternal, socioeconomic, and prenatal factors) does not necessarily hold. We need the nuisance coefficients which controls the impact of nuisance covariates on health indicator to be sparse, therefore, we can pick covariates of our choice from the real data that are shown to have a significant impact on the mortality rate of infant \citep{sabaphillip}. Furthermore, we can add some random correlated noise as extra covariates.
In this section, we examine the empirical performance of the proposed method using real data from the Natality birthweight dataset \citep{natality}. Our objective is to provide valid inference regarding the influence of mothers' smoking habits on infant mortality. Since the true impact of mothers' smoking habits on infant mortality is unknown, we simulate binary responses based on the covariates from the dataset to facilitate the calculation and comparison of coverage and bias. We generate these responses using information from studies on the impact of various maternal factors on infant mortality rate \citep{sabaphillip}. The generated output serves as a health indicator rather than the actual mortality rate and allows us to control the influence of mothers' smoking habits on the output.

In this synthesized output scenario, we incorporate the mother's smoking covariate from the real data. It's important to note that the working model \eqref{gamma_working} between the target variable (mothers' smoking in this case) and nuisance covariates (maternal, socioeconomic, and prenatal factors) might not strictly apply. 

We aim for sparse nuisance coefficients that control the impact of nuisance covariates on the health indicator. We select a few relevant covariates from the real data to have non-zero coefficients that have been demonstrated to significantly affect infant mortality rates \citep{sabaphillip}. Additionally, we introduce some randomly correlated noise as extra covariates to the model for a comprehensive analysis.
 
We generate our output based on the following model:
\begin{equation}
    P[Y_i=1 \mid X_i,\textbf{Z}_i] = \logit(\theta_* X_i + \beta_*^T \textbf{Z}_i + \mathbf{0}^T\mathbf{N_i}).
\end{equation}
Here $Y_i$ is the health indicator, $X_i$ is the mothers' smoking covariate and $\mathbf{Z_i}$ is 79 dimensional vector chosen from the birthweight dataset which consists of maternal factors, birth-related covariates, and $\mathbf{N_i}$ are some extra correlated noise generated according to $\mathbf{N_i} \sim N(0,H_{d-80 \times d-80})$ with $H_{ij} = 0.5^{|i-j|}$. $\beta_*$ is the sparse nuisance parameter that we choose based on the results discussed in \citep{sabaphillip}. To be more specific, the non-zero entries of $\beta_*$ correspond to factors such as the infant birthweight, BMI of the mother, obesity status of the mother, and age. In total, 4 elements of $\beta_*$ are non-zero. We choose $\theta_*$ from $\{-0.25, -0.4\}$ and perform Monte Carlo simulations. In each iteration, we draw a sub-sample of size $n = 500$ and generate a health indicator based on the model described above. We fit the model using the methods discussed in Section~\ref{methods} except BLASSO (which we did not include due to computational issues). Based on 1000 Monte Carlo results, we compare $95 \%$ coverage, interval length, and bias. %associated to the methods under consideration
\subsection{Discussion of results}

\begin{table}[ht]
\centering
\resizebox{\linewidth}{!}{%
\begin{tabular}{|c|c|c|c|c|c|c|c|c|}
  \hline
  \multicolumn{2}{|c|}{}& \multicolumn{3}{|c|}{Frequentist}& \multicolumn{2}{|c|}{Bayesian} & \multicolumn{1}{|c|}{}\\
  \hline
Quantities & $\theta_*$ & NAIVE & WLP & LSW & BMA & CB & ORACLE \\ 
  \hline
Coverage & -0.25 & 0.925 & 0.488 & 0.959 & 0.000 & 0.959 & 0.945 \\ 
   & -0.4 & 0.924 & 0.480 & 0.952 & 0.000 & 0.955 & 0.953 \\ 
   \hline
\hline
Interval & -0.25 & 1.054 & 1.013 & 1.702 & 0.004 & 1.237 & 0.990 \\ 
  Length & -0.4 & 1.054 & 1.589 & 1.705 & 0.001 & 1.226 & 0.992 \\ 
   \hline
\hline
Bias & -0.25 & 0.054 & 0.337 & 0.039 & 0.250 & 0.088 & 0.012 \\ 
   & -0.4 & 0.088 & 0.339 & 0.065 & 0.400 & 0.093 & 0.011 \\ 
   \hline
\end{tabular}}
\caption{Coverage, interval length, and bias corresponding to sample-based standard errors for each method considered for $n=300, d=400$ under signal strengths $\theta_* \in \{-0.25,-0.4\}$.}
\label{synth_300}
\end{table}

\begin{table}[ht]
\centering
\resizebox{\linewidth}{!}{%
\begin{tabular}{|c|c|c|c|c|c|c|c|c|}
  \hline
  \multicolumn{2}{|c|}{}& \multicolumn{3}{|c|}{Frequentist}& \multicolumn{2}{|c|}{Bayesian} & \multicolumn{1}{|c|}{}\\
 \hline
Quantities & $\theta_*$ & NAIVE & WLP & LSW & BMA & CB & ORACLE \\ 
  \hline
Coverage & -0.25 & 0.922 & 0.387 & 0.949 & 0.001 & 0.955 & 0.944 \\ 
   & -0.4 & 0.903 & 0.380 & 0.953 & 0.000 & 0.960 & 0.954 \\ 
   \hline
\hline
Interval & -0.25 & 0.906 & 0.987 & 1.488 & 0.003 & 1.036 & 0.852 \\ 
  Length & -0.4 & 0.907 & 1.248 & 1.477 & 0.003 & 1.038 & 0.855 \\ 
   \hline
\hline
Bias & -0.25 & 0.066 & 0.340 & 0.056 & 0.250 & 0.083 & 0.005 \\ 
   & -0.4 & 0.057 & 0.332 & 0.056 & 0.400 & 0.079 & 0.011 \\ 
   \hline
\end{tabular}}
\caption{Coverage, interval length, and bias corresponding to sample-based standard errors for each method considered for $n=400, d=500$ under signal strengths $\theta_* \in \{-0.25,-0.4\}$.}
\label{synth_400}
\end{table}

In this section, we compare the results for all the methods considered. Table~\ref{synth_300} and Table~\ref{synth_400} contain coverage, length of the interval, and bias obtained from the methods under consideration for the $(n=300,d=400)$ and $(n=400,d=500)$ cases, respectively. Based on Tables~\ref{synth_300}-\ref{synth_400}, we observe that the coverage attained by the proposed method (CB) is similar to the desired coverage. The length of the intervals and the bias of the proposed method are slightly larger than the optimal ones produced by ORACLE. On the other hand, NAIVE, WLP, and BMA suffer from under-coverage. The bias is much larger in the case of WLP and BMA. Among the frequentist methods, LSW attains the desired coverage with a small bias. When compared to the proposed method, LSW produces intervals with much larger lengths. In summary, the proposed method demonstrates a competitive overall performance in terms of achieving the desired coverage. 

\section{Real Data Analysis - Chronic Kidney Disease }\label{ckd_analysis}

We revisit the CKD data \citep{ckd_data} described briefly for the motivating example in Section~\ref{motivating_eg}. This dataset comprises information on 400 patients, with 246 of them diagnosed with Chronic Kidney Disease (CKD) and 135 diagnosed with diabetes. Each patient's record includes numerical features such as blood pressure, age, blood glucose level, blood urea level, etc. Additionally, the dataset incorporates categorical features like the presence of coronary heart disease, anemia, albumin levels, etc.

In the pre-processing phase, missing value-containing rows were removed. Consequently, the final dataset for analysis comprises 208 patients, out of which 91 have been diagnosed with Chronic Kidney Disease (CKD), and 56 have been diagnosed with diabetes. Categorical features were transformed into dummy variables, while numeric features were replaced with their respective spline basis functions. Additionally, interactions among all the nuisance covariates were taken into account. Consequently, our analysis incorporates a total of 2253 features, including the binary covariate indicating the presence of diabetes.

\subsection{Goal of the Study}

Our aim is to quantify the association between CKD and the presence of diabetes while controlling for all the other features in the analysis. Formally, we assume the following logistic model for $i=1,2, \dots, 208,$:
\begin{align}\label{ckd_eqn}
    P(\text{CKD} = 1 \mid \text{diabetes}, f_1, \dots, f_{2252}) = 
    (1 + \exp{-(\alpha +  \theta*\text{diabetes} + \sum_{j=1}^{2252} \alpha_j f_j)})^{-1}, 
\end{align}
where $\text{CKD} = 1$ and $\text{CKD}=0$ represent the presence and absence of the CKD, respectively. Similarly, $\text{diabetes} =1$ and $\text{diabetes} =0$ represent the presence and absence of the diabetes, respectively. All the nuisance features including the interactions are denoted by $f_1, f_2, \dots, f_{2252}$. Our aim is to estimate $\theta$ in Equation~\eqref{ckd_eqn} and provide a valid interval estimator for $\theta$.

\subsection{Results and Discussion}

\begin{table}[ht]
\centering
\begin{tabular}{|c|c|c|c|c|}
  \hline
Methods & Estimate for $\theta$& SE of the Estimate & LOWER & UPPER \\
  \hline
    CB & 0.693 & 0.392 & 0.191 & 1.195 \\ 
  BMA & 0.000 & 0.000 & 0.000 & 0.000 \\ 
  \hline
  LSW & -0.094 & 0.310 & -0.491 & 0.303 \\ 
  WLP & 0.451 & 0.221 & 0.168 & 0.734 \\ 
  NAIVE & 0.186 & 0.340 & -0.250 & 0.622 \\ 
  DS & 0.810 & 0.220 & 0.528 & 1.092 \\ 
   \hline
\end{tabular}
\caption{Estimates for $\theta$, standard errors (SE), lower and upper ends of interval estimators (90\% coverage) for CKD data analysis.}
\label{ckd_results}
\end{table}

 We compare the proposed method with the Bayesian Model Averaging (BMA), LSW, and WLP that are implemented as described in Section~\ref{simulation}. BLASSO was not included in this part of the study since it exhibited significantly longer convergence times. To facilitate a computationally efficient implementation of the proposed method, we obtain the nuisance estimates via LASSO and obtain the samples of the parameter of interest, $\theta$ via the proposed conditional posterior. This serves as a demonstration of effectiveness of the proposed conditional posterior in such settings when $d >>n$.

 Table~\ref{ckd_results} presents the results of our analysis of the association between diabetes and CKD. For each method under comparison, we report the estimate for $\theta$, the corresponding standard errors, lower ends, and the upper ends of the confidence intervals.

We observe that the inference obtained by the proposed method is strong and positive which corroborates the trend observed in the medical studies \citep{dm3,dm2,dm1}. However, other methods fail to detect this association.

\section{Extension to the Ordinal Treatment Setting}\label{ordinal_extension}

We assume that $X$ has $(K+1)$ categories such that $K$ is fixed, finite, and much smaller than the sample size $n$. We define the dummy covariates $X^1, X^2, \dots, X^K$ such that for $j=1, \dots, K$ and $i=1, \dots, n$:

\begin{equation}\label{dummy}
 X_i = (X^1_i,X^2_i,\dots,X^K_i), \text{ where }   X^j_i =
\begin{cases}
    1,& \text{if } X_i = j+1\\
    0,              & \text{otherwise}
\end{cases}
\end{equation}
Based on the definitions of the dummy covariates \eqref{dummy}, we assume the following logistic model for $i= 1, \dots, n$:

\begin{align}\label{logistic_categorical}
P[Y_i=1 \mid X_i,Z_i, \theta, \beta] &= \frac{\exp(X_i^T\theta + Z_i^T\beta)}{1+\exp(X_i^T\theta + Z_i^T\beta)},\nonumber \\
&= \frac{\exp(\sum_{j=1}^K X^j_i\theta^j + Z_i^T\beta)}{1+\exp(\sum_{j=1}^K X^j_i\theta^j + Z_i^T\beta)}.
\end{align}
In the case of categorical $X$, the parameter of interest is a $K$-dimensional vector $\theta = (\theta^1,\theta^2,\dots, \theta^K)$. We represent the corresponding oracle quantity as $\theta_0 = (\theta_0^1,\theta_0^2,\dots,\theta_0^K)$. For a $K-$dimensional vector $v= (v^1, v^2,\dots, v^K)$, we use the notation $v^{(-j)}$ to represent all the components of $v$ except $v^j$. The score function for $\theta^j$ is given by (for $j=1,\dots, K)$:
\begin{equation}\label{score_og_categorical}
    \varphi_{usual}(\theta^j \mid X,Y,Z,\theta^{(-j)},\beta) = \sum_{i=1}^n \frac{\partial}{\partial \theta^j}\log P[Y_i =1 \mid X_i,Z_i, \theta, \beta] = X^T(Y-\mu),
\end{equation}
where $\mu_{n \times 1}$ is the logistic mean vector with $\mu_i = \exp(X_i^T\theta + Z_i^T\beta)/(1 +\exp(X_i^T\theta + Z_i^T\beta) )$.
Consider $W_{n\times n}$ a diagonal matrix with diagonal elements defined as $W_{i,i} = \exp(X_i^T\theta_0 + Z_i^T\beta_0)/(1 +\exp(X_i^T\theta_0 + Z_i^T\beta_0) )^2$. Expected value of the derivative of this score \eqref{score_og_categorical} with respect to the nuisance parameter $\beta$, when evaluated at $\theta_0$ and $\beta_0$ leads to (for $j=1,\dots, K$):
\begin{equation}\label{usual_derivative2}
    E\bigg[\frac{\partial}{\partial \beta}\varphi_{usual}(\theta^j \mid X,Y,Z,\theta^{(-j)},\beta)\big|_{\theta_0,\beta_0}\bigg] = E[-(X^j)^TW Z]\neq 0.
\end{equation}
Therefore, the usual score for each component of the parameter of interest does not have the orthogonal property. Each dummy covariate $X^j$ can be treated similar to the binary covariate setting. With the same definition of $\mu$ as in \eqref{score_og_categorical}, consider the following score for each component $\theta^j$ of parameter of interest (for $j=1,2,\dots,K$):

\begin{equation}\label{score_proposal_categorical}
    \psi(\theta^j \mid X,Y,Z,\beta) = (X- h_0^j)^T(Y- \mu),
\end{equation}
% where $X^{(-j)}$ represents the matrix of dummy covariates except the column $X^j$ 
where $h_0^j:= (h_0(X_1^{(-j)},Z_1), h_0(X_2^{(-j)},Z_2), \dots, h_0(X_n^{(-j)},Z_n))$ is an extra function of the nuisance covariates and all the dummy covariates except $X^j$ that has been introduced to impose orthogonality. We desire the following property to achieve the orthogonality for this proposed score:
\begin{equation}\label{orthogonality_}
    E\big[\partial_{\beta}(\psi(\theta^j \mid X,Y,Z,\beta))\big|_{\theta_0,\beta_0}\big] =  0.
\end{equation}

The following expression for the function $h_0(X_i^{(-j)},Z_i)$ (for $i = 1,\dots, n$) leads to the desired orthogonal property for $\theta^j$:

\begin{align}\label{h_fn_categorical}
    h_0(X_i^{(-j)},Z_i) &=
    % \frac{P(Y_i=1\mid X_i^j=1, X^{(-j)},Z_i)P(Y_i=0\mid X_i^j=1, X^{(-j)}, Z_i)P(X_i^j=1\mid Z_i)}{\sum_{k \in \{0,1\}}P(Y_i=1\mid X_i^j=k, X^{(-j)}, Z_i)P(Y_i=0\mid X_i^j=k, X^{(-j)}, Z_i)P(X_i^j=k\mid Z_i)}, \nonumber\\
    % &=
    \frac{1}{1 + R_{0i}^j},\nonumber\\
    R_{0i}^j = &\dfrac{P(Y_i=1\mid X_i^j=0, X^{(-j)},Z_i)P(Y_i=0\mid X_i^j=0, X^{(-j)},Z_i)P(X_i^j=0\mid Z_i)}{P(Y_i=1\mid X_i^j=1, X^{(-j)},Z_i)P(Y_i=0\mid X_i^j=1, X^{(-j)},Z_i)P(X_i^j=1\mid Z_i)}.
\end{align}
where the quantities $P(Y_i=1\mid X_i^j=1, X^{(-j)},Z_i)$ and $P(Y_i=1\mid X_i^j=0, X^{(-j)},Z_i)$ are evaluated at the oracle values of the model parameters $(\theta_0,\beta_0)$. We have suppressed the use of extra notations in their respective conditional expression.
Given the observed data $(Y,X,Z)$, we do not have the knowledge of $P(Y_i = 1 \mid X_i, Z_i)$ and $P(X_i^j = 1 \mid Z_i)$. These two quantities are sufficient to obtain an estimator of $h_0^j$. 
Similar to the binary case, the expression for $h_0^j$ in \eqref{h_fn_categorical} suggests that it represents the conditional expectation of $X^j$, weighted by the variances of the binary output in the two sub-classes of $X^j$. Given this characteristic, we refer to the function $h_0^j$ as a variance-weighted projection. 

\subsection{Estimator of $P(Y_i = 1 \mid X_i, Z_i)$}
We follow similar steps as in the binary $X$ setting to obtain an estimator for $P(Y_i = 1 \mid X_i, Z_i)$. We use a sample of the parameters $(\Tilde{\theta}, \beta)$ to obtain estimators for $P(Y_i = 1 \mid X_i, Z_i)$. The working model for $(\tilde{\theta},\beta)$ is the same as the original logistic model given by \eqref{logistic_categorical}. Consequently, we obtain samples of $(\tilde{\theta},\beta)$ based on the following conditional posterior:
\begin{equation}\label{beta_posterior_categorical}
\tag{$(\tilde{\theta}, \beta)$-\textbf{CP}}
\boxed{
    f((\tilde{\theta}, \beta) \mid Y,X, Z) \propto \pi(\tilde{\theta})\pi( \beta)\prod_{i=1}^{n}\frac{\exp(X_i^T\tilde{\theta} + Z_i^T\beta)^{
    Y_i}}{1 + \exp(X_i^T\tilde{\theta}+ Z_i^T\beta)},}
\end{equation}
where $\pi(\tilde{\theta})$ and $\pi( \beta)$ are the priors on the model parameters $\theta$ and $\beta$, respectively. The tilde notation in $\Tilde{\theta}$ emphasizes that this sample is different from the final sample used for inference on the parameter of interest defined later by Equation \eqref{theta_posterior}.

The estimators for $P(Y_i = 1 \mid X_i, Z_i)$ corresponding to $X_i^j=1$ and $X_i^j = 0$, respectively are defined as:
\begin{align}\label{p_y_def_categorical}
    P(Y_i = 1 \mid X_i^j = 1, X^{(-j)}, Z_i, \tilde{\theta},\beta) &= \frac{\exp(\tilde{\theta^j} + (X^{(-j)})^T\tilde{\theta}^{(-j)}+ Z_i^T\beta)}{1+\exp(\tilde{\theta^j} + (X^{(-j)})^T\tilde{\theta}^{(-j)}+ Z_i^T\beta)} \nonumber \text{, and } \\
    P(Y_i = 1 \mid X_i = 0, X^{(-j)}, Z_i,\tilde{\theta},\beta) &= \frac{\exp((X^{(-j)})^T\tilde{\theta}^{(-j)}+ Z_i^T\beta)}{1+\exp((X^{(-j)})^T\tilde{\theta}^{(-j)}+ Z_i^T\beta)}.
\end{align}
It is worth noting that the estimators in \eqref{p_y_def_categorical} can be calculated simultaneously for all the dummy covariates based on the sample of $(\tilde{\theta},\beta)$.

\subsection{Estimator for the propensity scores: $P(X^j_i=1 \mid Z_i)$}

We introduce another set of nuisance parameters $\gamma^j$ (for $j=1, \dots, K$) to model the dependence between $X^j$ and $Z$ (for $j=1 \dots, K$, respectively) and therefore, obtain an estimator for $P(X_i^j=1 \mid Z_i)$. The working model for parameter $\gamma^j$ is given by the following logistic model for $i = 1, \dots, n$:

\begin{equation}\label{gamma_working_categorical}
\tag{$\gamma^j$-\textbf{WM}}
    P(X_i^j = 1 \mid Z_i, \gamma^j) = \frac{\exp( Z_i^T\gamma^j)}{1+\exp(Z_i^T\gamma^j)}.
\end{equation}
Motivated by\eqref{gamma_working_categorical}, we sample  $\gamma_j$ based on the following conditional posterior:

\begin{equation}\label{gamma-posterior-categorical}
\tag{$\gamma^j$-\textbf{CP}}
\boxed{
    f(\gamma^j \mid X,Z) \propto \pi(\gamma^j)\prod_{i=1}^{n}\frac{\exp(Z_i^T\gamma^j)^{
    X_i}}{1 + \exp(Z_i^T\gamma^j)},}
\end{equation}
where $\pi(\gamma^j)$ is prior on the parameter $\gamma^j$. Based on a sample of $\gamma$ from \eqref{gamma-posterior-categorical}, the estimators for the propensity scores are given by the expression of the working model \eqref{gamma_working_categorical}.

\subsection{Estimator of the variance-weighted projection}

Based on a sample of $(\tilde{\theta}, \beta)$ from \eqref{beta_posterior}, we can obtain the estimators for $P(Y_i=1 \mid X_i^j=1, X_i^{(-j}, Z_i)$ and $P(Y_i=1 \mid X_i^j=0, X_i^{(-j},Z_i)$ as defined in Equation~\eqref{p_y_def_categorical}. Moreover, based on a sample of $\gamma^j$ from \eqref{gamma-posterior-categorical}, we can obtain an estimator for $P(X_i^j =1 \mid Z_i)$ as defined through the working model \eqref{gamma_working_categorical}. For $j= 1, \dots, K$, we can obtain an estimator for $h_0^j$ which is denoted by $h^j$ and is defined as (for $i=1,\dots,n$):

\begin{align}\label{h_estimate_categorical}
    h^j_i &= \frac{1}{1+ R_i^j},
    \nonumber\\
    R_i^j = &\dfrac{P(Y_i=1\mid X_i^j=0, X_i^{(-j)},Z_i, \tilde{\theta},\beta)P(Y_i=0\mid X_i^j=0, X_i^{(-j)},Z_i, \tilde{\theta},\beta)P(X_i^j=0\mid Z_i, \gamma^j)}{P(Y_i=1\mid X_i^j=1, X_i^{(-j)}, Z_i,\tilde{\theta},\beta)P(Y_i=0\mid X_i^j=1, X_i^{(-j)}, Z_i,\tilde{\theta},\beta)P(X_i^j=1\mid Z_i,\gamma^j)}.
\end{align}
In summary, a sample for the estimator of the variance-weighted projection is obtained from a sample of the model parameters $(\tilde{\theta}, \beta)$, which is drawn from \eqref{beta_posterior_categorical}, and samples of $\{\gamma^1,\gamma^2, \dots,\gamma^K\}$ drawn from \eqref{gamma-posterior-categorical} (for $j=1,\dots, K)$.

\subsection{Working model for the parameter of interest $\theta$}

We first define a re-parameterized nuisance quantity based on the samples $(\tilde{\theta},\beta)$ \eqref{beta_posterior_categorical} and the estimator $h_i$ defined as $h_i = (h_i^1, h_i^2, \dots, h_i^K)$ \eqref{h_estimate_categorical} given by:

\begin{equation}\label{phi_def_categorical}
   \phi = (\phi_1, \dots, \phi_n), \text{ where } \phi_i = \sum_{j=1}^K\tilde{\theta}^jh_i^j + Z_i^T\beta = h_i^T\tilde{\theta} + Z_i^T\beta, \text{ for } i=1, \dots,n.
\end{equation}
Note that we use $\tilde{\theta} $ in the estimation of this quantity $\phi$ as it is sampled based on \eqref{beta_posterior} and is different from the conditional posterior for $\theta$ used for final inference. This makes the posterior concentration for $\phi$ and $\theta$ not depend on each other.
We assume the following working model for the parameter of interest $\theta$ conditional on $h(Z)$  and $\phi$:

\begin{equation}\label{theta_working_categorical}
\tag{$\theta$-\textbf{WM}}
    P(Y_i=1 \mid X_i,Z_i, h_i,\phi_i, \theta) = \frac{\exp\{(X_i - h_i)^T\theta + \phi_i\}}{1 + \exp\{(X_i - h_i)^T\theta + \phi_i\} },
\end{equation}
for $i=1, \dots, n$. 
The score function for $\theta^j$ corresponding to \eqref{theta_working_categorical} has the orthogonal property w.r.t. $h^j$ and $\phi$ for each $j =1, 2, \dots, K$. This fact can be established based on the definition of the variance-weighted projection \eqref{h_fn_categorical} and the original model assumption \eqref{logistic_categorical}.

\subsubsection{Conditional Posterior for $\theta$:}

Based on the working model \eqref{theta_working}, we obtain a sample for the parameter of interest $\theta$ based on the following conditional posterior:
\begin{equation}\label{theta_posterior_categorical}
    f(\theta \mid Y_i, X_i, Z_i, h_i, \phi_i) \propto \pi(\theta) \prod_{i=1}^{n}\frac{\exp\{(X_i - h_i)^T\theta + \phi_i\}^{Y_i}}{1 + \exp\{(X_i - h_i)^T\theta + \phi_i\} }.
\end{equation}
where $\pi(\theta)$ is the prior distribution on the parameter of interest $\theta.$

The sampling scheme and the theoretical results can easily be extended to the case of categorical $X$. 

\section{A generalized expression of variance-weighted projection}\label{general_form}

We assume the following logistic model for $i= 1, \dots, n$:

\begin{equation*}
    P[Y_i=1 \mid X_i,Z_i, \theta, \beta] = \frac{\exp(X_i^T\theta + Z_i^T\beta)}{1+\exp(X_i^T\theta + Z_i^T\beta)}.
\end{equation*}

We consider the following proposal for an orthogonal score using a variance-weighted projection $h_0(Z)$: 

\begin{equation*}\label{score_proposal_general}
    \psi(\theta \mid X,Y,Z,\beta) = (X- h_0(Z))^T(Y- \mu),
\end{equation*}
where $h_0(Z):= (h_0(Z_1), h_0(Z_2), \dots, h_0(Z_n))$. We desire the following property to achieve the orthogonality for this proposed score:
\begin{equation}\label{orthogonality_general}
    E\big[\partial_{\beta}(\psi(\theta \mid X,Y,Z,\beta))\big|_{\theta_0,\beta_0}\big] =  0.
\end{equation}
Upon expanding the expected derivative in \eqref{orthogonality_general}, we obtain the expression:

\begin{equation*}
    E\bigg[\frac{\partial}{\partial \beta}\psi(\theta| X,Y,Z,\beta)\big|_{\theta_0,\beta_0}\bigg] = -\sum_{i=1}^n E\big[(X_i - h_0(Z_i))\mu_{0i}(1-\mu_{0i})\big],
\end{equation*}
where $\mu_{0i} = P(Y_i =1 \mid X_i, Z_i, \theta_0, \beta_0)$.
Therefore, to attain the orthogonal property from \eqref{orthogonality_general}, it is sufficient to satisfy the following equation by the tower property of conditional expectations:
\begin{equation}\label{h0_derivation_general}
E\big[ (X_i - h_0(Z_i))\mu_{0i}(1-\mu_{0i})\mid Z_i\big] = 0.
\end{equation}
Consequently, we can obtain a general expression for the variance-weighted projection $h_0(Z)$ based on \eqref{h0_derivation_general} as follows:

\begin{equation}\label{h_0_general}
    h_0(Z_i) = \frac{E\big[ X_i Var(Y_i \mid X_i, Z_i) \mid Z_i\big]}{E\big[ Var(Y_i \mid X_i, Z_i) \mid Z_i\big]}.
\end{equation}

In the specific setting where $X_i \in \{0,1, 2, \dots, K\}$ with $K$ being finite and small positive integer, we obtain the following explicit expression for the variance-weighted projection $h_0(Z)$:

\begin{equation}
    h_0(Z_i) = \frac{\sum_{k=1}^K k Var(Y_i \mid X_i = k, Z_i) P(X_i = k \mid Z_i)}{\sum_{k=1}^K Var(Y_i \mid X_i = k, Z_i) P(X_i = k \mid Z_i)}.
\end{equation}

In practice, $Var(Y_i \mid X_i, Z_i)$ and $P(X_i \mid Z_i)$ are unknown quantities. Therefore, we need to devise ways to obtain estimators for these quantities. In general settings, obtaining a closed-form expression for the conditional expectations in \eqref{h_0_general} can be difficult. One can attempt to use other modeling techniques to obtain an estimator for these conditional expectations.

\section{Conclusion}
\label{conclusion}

\bibliographystyle{asa}

\bibliography{references}

\begin{thebibliography}{82}
\newcommand{\enquote}[1]{``#1''}
\expandafter\ifx\csname natexlab\endcsname\relax\def\natexlab#1{#1}\fi

\bibitem[{Agresti(2012)}]{agresti2012categorical}
Agresti, A. (2012), \textit{Categorical data analysis}, vol. 792, John Wiley \& Sons.

\bibitem[{Antonelli et~al.(2022)Antonelli, Papadogeorgou, and Dominici}]{antonelli2022causal}
Antonelli, J., Papadogeorgou, G., and Dominici, F. (2022), \enquote{Causal inference in high dimensions: a marriage between Bayesian modeling and good frequentist properties,} \textit{Biometrics}, 78, 100--114.

\bibitem[{Armagan et~al.(2013)Armagan, Dunson, and Lee}]{armagan2013generalized}
Armagan, A., Dunson, D.~B., and Lee, J. (2013), \enquote{Generalized double {P}areto shrinkage,} \textit{Statistica Sinica}, 23, 119.

\bibitem[{Bach(2010)}]{bach2010self}
Bach, F. (2010), \enquote{Self-concordant analysis for logistic regression,} \textit{Electronic Journal of Statistics}, 4, 384--414.

\bibitem[{Belloni et~al.(2012)Belloni, Chen, Chernozhukov, and Hansen}]{belloni2012sparse}
Belloni, A., Chen, D., Chernozhukov, V., and Hansen, C. (2012), \enquote{Sparse models and methods for optimal instruments with an application to eminent domain,} \textit{Econometrica}, 80, 2369--2429.

\bibitem[{Belloni et~al.(2010)Belloni, Chernozhukov, and Hansen}]{belloni2010lasso}
Belloni, A., Chernozhukov, V., and Hansen, C. (2010), \enquote{Lasso methods for {G}aussian instrumental variables models,} \textit{arXiv preprint arXiv:1012.1297}.

\bibitem[{Belloni et~al.(2014)Belloni, Chernozhukov, and Hansen}]{belloni2014inference}
--- (2014), \enquote{Inference on treatment effects after selection among high-dimensional controls,} \textit{The Review of Economic Studies}, 81, 608--650.

\bibitem[{Belloni et~al.(2013)Belloni, Chernozhukov, and Wei}]{logistic_dml}
Belloni, A., Chernozhukov, V., and Wei, Y. (2013), \enquote{{Honest confidence regions for a regression parameter in logistic regression with a large number of controls},} CeMMAP working papers CWP67/13, Centre for Microdata Methods and Practice, Institute for Fiscal Studies.

\bibitem[{Berger et~al.(1999)Berger, Liseo, and Wolpert}]{berger1999integrated}
Berger, J.~O., Liseo, B., and Wolpert, R.~L. (1999), \enquote{Integrated likelihood methods for eliminating nuisance parameters,} \textit{Statistical science}, 1--22.

\bibitem[{Berk et~al.(2013)Berk, Brown, Buja, Zhang, and Zhao}]{berk2013valid}
Berk, R., Brown, L., Buja, A., Zhang, K., and Zhao, L. (2013), \enquote{Valid post-selection inference,} \textit{The Annals of Statistics}, 802--837.

\bibitem[{Bhattacharya et~al.(2015)Bhattacharya, Pati, Pillai, and Dunson}]{bhattacharya2015dirichlet}
Bhattacharya, A., Pati, D., Pillai, N.~S., and Dunson, D.~B. (2015), \enquote{Dirichlet--{L}aplace priors for optimal shrinkage,} \textit{Journal of the American Statistical Association}, 110, 1479--1490.

\bibitem[{Bondell and Reich(2012)}]{bondell2012consistent}
Bondell, H.~D. and Reich, B.~J. (2012), \enquote{Consistent high-dimensional {B}ayesian variable selection via penalized credible regions,} \textit{Journal of the American Statistical Association}, 107, 1610--1624.

\bibitem[{Breheny and Huang(2011)}]{breheny2011coordinate}
Breheny, P. and Huang, J. (2011), \enquote{Coordinate descent algorithms for nonconvex penalized regression, with applications to biological feature selection,} \textit{The Annals of Applied Statistics}, 5, 232.

\bibitem[{Bunea(2008)}]{bunea2008honest}
Bunea, F. (2008), \enquote{Honest variable selection in linear and logistic regression models via $l_1$ and $l_1$+ $l_2$ penalization,} \textit{Electronic Journal of Statistics}, 2, 1153--1194.

\bibitem[{Cai and Guo(2017)}]{cai2017confidence}
Cai, T.~T. and Guo, Z. (2017), \enquote{Confidence intervals for high-dimensional linear regression: Minimax rates and adaptivity,} .

\bibitem[{Cai et~al.(2021)Cai, Guo, and Ma}]{other}
Cai, T.~T., Guo, Z., and Ma, R. (2021), \enquote{Statistical inference for high-dimensional generalized linear models with binary outcomes,} \textit{Journal of the American Statistical Association}, 1--14.

\bibitem[{Candes and Tao(2007)}]{candes2007dantzig}
Candes, E. and Tao, T. (2007), \enquote{The Dantzig selector: Statistical estimation when p is much larger than n,} .

\bibitem[{Carvalho et~al.(2009)Carvalho, Polson, and Scott}]{carvalho2009handling}
Carvalho, C.~M., Polson, N.~G., and Scott, J.~G. (2009), \enquote{Handling sparsity via the horseshoe,} in \textit{Artificial Intelligence and Statistics}, PMLR, pp. 73--80.

\bibitem[{Castillo et~al.(2015)Castillo, Schmidt-Hieber, and Van~der Vaart}]{castillo2015Bayesian}
Castillo, I., Schmidt-Hieber, J., and Van~der Vaart, A. (2015), \enquote{{B}ayesian linear regression with sparse priors,} \textit{The Annals of Statistics}, 43, 1986--2018.

\bibitem[{CDC(2020)}]{diabetes_report}
CDC (2020), \textit{National Diabetes Statistics Report, 2020}, Atlanta, GA: Centers for Disease Control and Prevention, U.S. Dept of Health and Human Services.

\bibitem[{Chernozhukov et~al.(2018)Chernozhukov, Chetverikov, Demirer, Duflo, Hansen, Newey, and Robins}]{dml}
Chernozhukov, V., Chetverikov, D., Demirer, M., Duflo, E., Hansen, C., Newey, W., and Robins, J. (2018), \enquote{Double/debiased machine learning for treatment and structural parameters,} \textit{The Econometrics Journal}, 21, C1--C68.

\bibitem[{Cochran and Chambers(1965)}]{cochran1965planning}
Cochran, W.~G. and Chambers, S.~P. (1965), \enquote{The planning of observational studies of human populations,} \textit{Journal of the Royal Statistical Society. Series A (General)}, 128, 234--266.

\bibitem[{Collett(2002)}]{collett2002modelling}
Collett, D. (2002), \textit{Modelling binary data}, CRC press.

\bibitem[{Fahrmeir and Kaufmann(1985)}]{fahrmeir1985consistency}
Fahrmeir, L. and Kaufmann, H. (1985), \enquote{Consistency and asymptotic normality of the maximum likelihood estimator in generalized linear models,} \textit{The Annals of Statistics}, 13, 342--368.

\bibitem[{Fan and Li(2001)}]{scad}
Fan, J. and Li, R. (2001), \enquote{Variable selection via nonconcave penalized likelihood and its oracle properties,} \textit{Journal of the American Statistical Association}, 96, 1348--1360.

\bibitem[{Friedman et~al.(2010)Friedman, Hastie, and Tibshirani}]{friedman2010regularization}
Friedman, J., Hastie, T., and Tibshirani, R. (2010), \enquote{Regularization paths for generalized linear models via coordinate descent,} \textit{Journal of Statistical Software}, 33, 1.

\bibitem[{Genkin et~al.(2007)Genkin, Lewis, and Madigan}]{genkin2007large}
Genkin, A., Lewis, D.~D., and Madigan, D. (2007), \enquote{Large-scale {B}ayesian logistic regression for text categorization,} \textit{{T}echnometrics}, 49, 291--304.

\bibitem[{George and McCulloch(1993)}]{george1993variable}
George, E.~I. and McCulloch, R.~E. (1993), \enquote{Variable selection via {G}ibbs sampling,} \textit{Journal of the American Statistical Association}, 88, 881--889.

\bibitem[{Ghosh et~al.(2018)Ghosh, Li, and Mitra}]{ghosh2018use}
Ghosh, J., Li, Y., and Mitra, R. (2018), \enquote{On the use of {C}auchy prior distributions for {B}ayesian logistic regression,} \textit{{B}ayesian Analysis}, 13, 359--383.

\bibitem[{Hahr and Molitch(2015)}]{dm2}
Hahr, A.~J. and Molitch, M.~E. (2015), \enquote{Management of diabetes mellitus in patients with chronic kidney disease,} \textit{Clinical diabetes and endocrinology}, 1, 1--9.

\bibitem[{Hosmer~Jr et~al.(2013)Hosmer~Jr, Lemeshow, and Sturdivant}]{hosmer2013applied}
Hosmer~Jr, D.~W., Lemeshow, S., and Sturdivant, R.~X. (2013), \textit{Applied logistic regression}, vol. 398, John Wiley \& Sons.

\bibitem[{Huang and Zhang(2012)}]{huang2012estimation}
Huang, J. and Zhang, C.-H. (2012), \enquote{Estimation and selection via absolute penalized convex minimization and its multistage adaptive applications,} \textit{The Journal of Machine Learning Research}, 13, 1839--1864.

\bibitem[{Imbens and Rubin(2015)}]{imbens2015causal}
Imbens, G.~W. and Rubin, D.~B. (2015), \textit{Causal inference in statistics, social, and biomedical sciences}, Cambridge University Press.

\bibitem[{Ishwaran and Rao(2005)}]{israo}
Ishwaran, H. and Rao, J.~S. (2005), \enquote{{Spike and slab variable selection: Frequentist and {B}ayesian strategies},} \textit{The Annals of Statistics}, 33, 730 -- 773.

\bibitem[{Javanmard and Montanari(2014)}]{javanmard2014confidence}
Javanmard, A. and Montanari, A. (2014), \enquote{Confidence intervals and hypothesis testing for high-dimensional regression,} \textit{The Journal of Machine Learning Research}, 15, 2869--2909.

\bibitem[{Johnson and Rossell(2012)}]{johnson2012Bayesian}
Johnson, V.~E. and Rossell, D. (2012), \enquote{{B}ayesian model selection in high-dimensional settings,} \textit{Journal of the American Statistical Association}, 107, 649--660.

\bibitem[{Kabaila(1995)}]{kabaila1995effect}
Kabaila, P. (1995), \enquote{The effect of model selection on confidence regions and prediction regions,} \textit{Econometric Theory}, 11, 537--549.

\bibitem[{Kleinbaum et~al.(2002)Kleinbaum, Klein, and Pryor}]{kleinbaum2002logistic}
Kleinbaum, D.~G., Klein, M., and Pryor, E.~R. (2002), \textit{Logistic regression: a self-learning text}, vol.~94, Springer.

\bibitem[{Kuchibhotla et~al.(2020)Kuchibhotla, Brown, Buja, Cai, George, and Zhao}]{uposi}
Kuchibhotla, A.~K., Brown, L.~D., Buja, A., Cai, J., George, E.~I., and Zhao, L.~H. (2020), \enquote{{Valid post-selection inference in model-free linear regression},} \textit{The Annals of Statistics}, 48, 2953 -- 2981.

\bibitem[{Kwemou(2016)}]{kwemou2016non}
Kwemou, M. (2016), \enquote{Non-asymptotic oracle inequalities for the Lasso and group Lasso in high dimensional logistic model,} \textit{ESAIM: Probability and Statistics}, 20, 309--331.

\bibitem[{Leeb and P{\"o}tscher(2005)}]{leeb2005model}
Leeb, H. and P{\"o}tscher, B.~M. (2005), \enquote{Model selection and inference: Facts and fiction,} \textit{Econometric Theory}, 21, 21--59.

\bibitem[{Leeb and P{\"o}tscher(2008)}]{leeb2008sparse}
--- (2008), \enquote{Sparse estimators and the oracle property, or the return of {H}odges’ estimator,} \textit{Journal of Econometrics}, 142, 201--211.

\bibitem[{Liang et~al.(2008)Liang, Paulo, Molina, Clyde, and Berger}]{liang2008mixtures}
Liang, F., Paulo, R., Molina, G., Clyde, M.~A., and Berger, J.~O. (2008), \enquote{Mixtures of g-priors for {B}ayesian variable selection,} \textit{Journal of the American Statistical Association}, 103, 410--423.

\bibitem[{Ma et~al.(2021)Ma, Tony~Cai, and Li}]{ht}
Ma, R., Tony~Cai, T., and Li, H. (2021), \enquote{Global and simultaneous hypothesis testing for high-dimensional logistic regression models,} \textit{Journal of the American Statistical Association}, 116, 984--998.

\bibitem[{Martin et~al.(2017)Martin, Hamilton, Osterman, and et~al.}]{natality}
Martin, J.~A., Hamilton, B.~E., Osterman, M. J.~K., and et~al. (2017), \enquote{Births: Final Data for 2015,} \textit{National Vital Statistics Report}, 66.

\bibitem[{Masho and Archer(2011)}]{sabaphillip}
Masho, S.~W. and Archer, P.~W. (2011), \enquote{Does maternal birth outcome differentially influence the occurrence of infant death among {A}frican {A}mericans and {E}uropean {A}mericans?} \textit{Maternal and child health journal}, 15, 1249--1256.

\bibitem[{Narisetty and He(2014)}]{basad}
Narisetty, N.~N. and He, X. (2014), \enquote{{{B}ayesian variable selection with shrinking and diffusing priors},} \textit{The Annals of Statistics}, 42, 789 -- 817.

\bibitem[{Narisetty et~al.(2019)Narisetty, Shen, and He}]{skinny}
Narisetty, N.~N., Shen, J., and He, X. (2019), \enquote{Skinny {G}ibbs: A Consistent and Scalable {G}ibbs Sampler for Model Selection,} \textit{Journal of the American Statistical Association}, 114, 1205--1217.

\bibitem[{Neyman(1959)}]{neyman1959optimal}
Neyman, J. (1959), \enquote{Optimal asymptotic tests of composite hypotheses,} \textit{Probability and {S}tatistics}, 213--234.

\bibitem[{Neyman(1979)}]{neyman1979c}
--- (1979), \enquote{C ($\alpha$) tests and their use,} \textit{Sankhy{\=a}: The Indian Journal of Statistics, Series A}, 1--21.

\bibitem[{Nordheim and Jenssen(2021)}]{dm1}
Nordheim, E. and Jenssen, T.~G. (2021), \enquote{Chronic kidney disease in patients with diabetes mellitus,} \textit{Endocrine Connections}, 10, R151--R159.

\bibitem[{O'brien and Dunson(2004)}]{o2004bayesian}
O'brien, S.~M. and Dunson, D.~B. (2004), \enquote{Bayesian multivariate logistic regression,} \textit{Biometrics}, 60, 739--746.

\bibitem[{Ojha and Narisetty(2023)}]{ojha2023conditional}
Ojha, A. and Narisetty, N.~N. (2023), \enquote{A Conditional Bayesian Approach with Valid Inference for High Dimensional Logistic Regression,} \textit{Bayesian Analysis}, 1, 1--27.

\bibitem[{Panigrahi and Taylor(2018)}]{panigrahi2018scalable}
Panigrahi, S. and Taylor, J. (2018), \enquote{Scalable methods for {B}ayesian selective inference,} \textit{Electronic Journal of Statistics}, 12, 2355--2400.

\bibitem[{Panigrahi et~al.(2021)Panigrahi, Taylor, and Weinstein}]{panigrahi2021integrative}
Panigrahi, S., Taylor, J., and Weinstein, A. (2021), \enquote{Integrative methods for post-selection inference under convex constraints,} \textit{The Annals of Statistics}, 49, 2803--2824.

\bibitem[{Park and Casella(2008)}]{park2008Bayesian}
Park, T. and Casella, G. (2008), \enquote{The {B}ayesian lasso,} \textit{Journal of the American Statistical Association}, 103, 681--686.

\bibitem[{Persson and Waernbaum(2013)}]{conditional_odds_ratio}
Persson, E. and Waernbaum, I. (2013), \enquote{Estimating a marginal causal odds ratio in a case-control design: analyzing the effect of low birth weight on the risk of type 1 diabetes mellitus,} \textit{Statistics in Medicine}, 32, 2500--2512.

\bibitem[{Polson et~al.(2012)Polson, Scott, and Windle}]{polya-gamma}
Polson, N., Scott, J., and Windle, J. (2012), \enquote{{B}ayesian Inference for Logistic Models Using {P}olya-{G}amma Latent Variables,} \textit{Journal of the American Statistical Association}, 108.

\bibitem[{P{\"o}tscher(1991)}]{potscher1991effects}
P{\"o}tscher, B.~M. (1991), \enquote{Effects of model selection on inference,} \textit{Econometric Theory}, 7, 163--185.

\bibitem[{P{\"o}tscher(2009)}]{potscher2009confidence}
--- (2009), \enquote{Confidence sets based on sparse estimators are necessarily large,} \textit{Sankhy{\=a}: The Indian Journal of Statistics, Series A (2008-)}, 1--18.

\bibitem[{P{\"o}tscher and Leeb(2009)}]{potscher2009distribution}
P{\"o}tscher, B.~M. and Leeb, H. (2009), \enquote{On the distribution of penalized maximum likelihood estimators: The LASSO, SCAD, and thresholding,} \textit{Journal of Multivariate Analysis}, 100, 2065--2082.

\bibitem[{Ro{\v{c}}kov{\'a} and George(2014)}]{rovckova2014emvs}
Ro{\v{c}}kov{\'a}, V. and George, E.~I. (2014), \enquote{EMVS: The {EM} approach to {B}ayesian variable selection,} \textit{Journal of the American Statistical Association}, 109, 828--846.

\bibitem[{Rosenbaum and Rosenbaum(2002)}]{rosenbaum2002overt}
Rosenbaum, P.~R. and Rosenbaum, P.~R. (2002), \textit{Overt bias in observational studies}, Springer.

\bibitem[{Rubin(1974)}]{rubin1974estimating}
Rubin, D.~B. (1974), \enquote{Estimating causal effects of treatments in randomized and nonrandomized studies.} \textit{Journal of educational Psychology}, 66, 688.

\bibitem[{Rubini et~al.(2015)Rubini, Soundarapandian, , and Eswaran}]{ckd_data}
Rubini, L., Soundarapandian, P., , and Eswaran, P. (2015), \enquote{{Chronic Kidney Disease},} {DOI}: https://doi.org/10.24432/C5G020.

\bibitem[{Schervish(2012)}]{schervish2012theory}
Schervish, M.~J. (2012), \textit{Theory of {S}tatistics}, Springer Science \& Business Media.

\bibitem[{Severini(2007)}]{severini2007integrated}
Severini, T.~A. (2007), \enquote{Integrated likelihood functions for non-Bayesian inference,} \textit{Biometrika}, 94, 529--542.

\bibitem[{Shi et~al.(2021)Shi, Song, Lu, and Li}]{shi_recursive_scores}
Shi, C., Song, R., Lu, W., and Li, R. (2021), \enquote{Statistical inference for high-dimensional models via recursive online-score estimation,} \textit{Journal of the American Statistical Association}, 116, 1307--1318.

\bibitem[{Song and Liang(2022)}]{song2022nearly}
Song, Q. and Liang, F. (2022), \enquote{Nearly optimal {B}ayesian shrinkage for high-dimensional regression,} \textit{Science China Mathematics}, 1--34.

\bibitem[{Taylor and Tibshirani(2015)}]{taylor2015statistical}
Taylor, J. and Tibshirani, R.~J. (2015), \enquote{Statistical learning and selective inference,} \textit{Proceedings of the National Academy of Sciences}, 112, 7629--7634.

\bibitem[{Thomas et~al.(2016)Thomas, Cooper, and Zimmet}]{dm3}
Thomas, M.~C., Cooper, M.~E., and Zimmet, P. (2016), \enquote{Changing epidemiology of type 2 diabetes mellitus and associated chronic kidney disease,} \textit{Nature Reviews Nephrology}, 12, 73--81.

\bibitem[{Thomas et~al.(2008)Thomas, Kanso, and Sedor}]{ckd}
Thomas, R., Kanso, A., and Sedor, J.~R. (2008), \enquote{Chronic kidney disease and its complications,} \textit{Primary care: Clinics in office practice}, 35, 329--344.

\bibitem[{Torrens-i Dinar{\`e}s et~al.(2021)Torrens-i Dinar{\`e}s, Papaspiliopoulos, and Rossell}]{torrens2021confounder}
Torrens-i Dinar{\`e}s, M., Papaspiliopoulos, O., and Rossell, D. (2021), \enquote{Confounder importance learning for treatment effect inference,} \textit{arXiv preprint arXiv:2110.00314}.

\bibitem[{T{\"u}chler(2008)}]{bma_tuchler}
T{\"u}chler, R. (2008), \enquote{Bayesian variable selection for logistic models using auxiliary mixture sampling,} \textit{Journal of Computational and Graphical Statistics}, 17, 76--94.

\bibitem[{Van~de Geer et~al.(2014)Van~de Geer, B{\"u}hlmann, Ritov, and Dezeure}]{van2014asymptotically}
Van~de Geer, S., B{\"u}hlmann, P., Ritov, Y., and Dezeure, R. (2014), \enquote{On asymptotically optimal confidence regions and tests for high-dimensional models,} \textit{The Annals of Statistics}, 42, 1166--1202.

\bibitem[{Van~de Geer(2008)}]{geerGLM}
Van~de Geer, S.~A. (2008), \enquote{High-dimensional generalized linear models and the lasso,} \textit{The Annals of Statistics}, 36, 614--645.

\bibitem[{Vansteelandt et~al.(2011)Vansteelandt, Bowden, Babanezhad, and Goetghebeur}]{causal_iv_clor}
Vansteelandt, S., Bowden, J., Babanezhad, M., and Goetghebeur, E. (2011), \enquote{On instrumental variables estimation of causal odds ratios,} .

\bibitem[{Walker(1969)}]{walker}
Walker, A.~M. (1969), \enquote{On the Asymptotic Behaviour of Posterior Distributions,} \textit{Journal of the Royal Statistical Society. Series B (Methodological)}, 31, 80--88.

\bibitem[{Wang et~al.(2020)Wang, He, and Xu}]{wang2020debiased}
Wang, J., He, X., and Xu, G. (2020), \enquote{Debiased inference on treatment effect in a high-dimensional model,} \textit{Journal of the American Statistical Association}, 115, 442--454.

\bibitem[{Wu et~al.(2023)Wu, N.~Narisetty, and Yang}]{qbayesian}
Wu, T., N.~Narisetty, N., and Yang, Y. (2023), \enquote{Statistical inference via conditional Bayesian posteriors in high-dimensional linear regression,} \textit{Electronic Journal of Statistics}, 17, 769--797.

\bibitem[{Zhang and Zhang(2014)}]{zhang2014confidence}
Zhang, C.-H. and Zhang, S.~S. (2014), \enquote{Confidence intervals for low dimensional parameters in high dimensional linear models,} \textit{Journal of the Royal Statistical Society: Series B (Statistical Methodology)}, 76, 217--242.

\bibitem[{Zhu et~al.(2019)Zhu, Zeng, Zhang, and Li}]{zhu2019estimating}
Zhu, A., Zeng, D., Zhang, P., and Li, L. (2019), \enquote{Estimating causal log-odds ratio using the case-control sample and its application in the pharmaco-epidemiology study,} \textit{Statistical methods in medical research}, 28, 2165--2178.

\end{thebibliography}

\newpage
\appendix

\section{Gibbs Sampler for Computation}\label{sampler}
The proposed Bayesian framework provides a conditional posterior distribution for the model parameters that can be used for inferential purposes. Markov Chain Monte Carlo techniques can be utilized for sampling from the Bayesian posterior. 

The posterior obtained by the logistic likelihood itself does not present a viable option for direct posterior sampling. We utilize the latent variable formulation given by P\'{o}lya-Gamma \citep{polya-gamma} to convert the posteriors into mixtures of Gaussian distributions which makes it easy to sample from the posterior distribution. 
We present the sampling steps for the nuisance parameters $(\tilde{\theta},\beta), \gamma$, and the parameter $\theta$ from their conditional posterior distributions. 
The proposed Gibbs Sampler takes the following steps:
\begin{enumerate}[label=S\arabic*]
\item \textbf{Sampler for the conditional of $(\tilde{\theta},\beta)$}:\\
Sampling of $(\tilde{\theta},\beta)$ from \eqref{beta_posterior} is based on a logistic posterior distribution. We use the P\'{o}lya-Gamma latent variables (denote by $\omega_{11}, \dots \omega_{1n}$) to accomplish this task. In the prior specification of $\beta$ \eqref{eta-prior}, the binary latent variables for spike and slab is denoted by $I_1$.
For $i = 1, 2, \dots, n$, the conditional distribution of P\'{o}lya-Gamma random variables is given by:
$$\omega_{1i} \mid (\tilde{\theta},\beta), X_i, Z_i,I_1) \sim \text{PG}(1, X_i\tilde{\theta} + Z_{i}^T\beta).$$
Define $\Omega_1 := \text{Diag}(\omega_{11}, \dots \omega_{1n}) $. For simplicity of presentation, consider another $n \times (d+1)$ matrix $D = [X,Z]$.
Then the conditional distribution for $(\tilde{\theta},\beta)$ is given by:   
    \begin{equation*}
        (\tilde{\theta},\beta) \mid (I_1,X,Z,Y,\Omega_1) \sim N(\mu, \Sigma^{-1})
        % (Z^T\Omega_1Z + D_{I_1})^{-1}(Z^T(Y-1/2)),(Z^T\Omega_1Z + D_{I_1})^{-1}\bigg),
    \end{equation*}
where $$\Sigma = (D^T\Omega_1D + \Lambda^{-1}),$$ $$\text{with }\Lambda_{1,1} = \lambda,\ \  \Lambda_{2:(d+1),2:(d+1)} = \text{Diag}(I_1 \tau_{1n}^{2} + (1-I_1)\tau_{0n}^{2}), \text{and}$$ $$\mu = \Sigma^{-1}(D^T(Y-1/2)).$$
Let $\phi_N(r,\mu,\sigma^2)$ denotes the standard Gaussian density function of variable $r$ with mean $\mu$ and variance $\sigma^2$. The conditional distribution for indicators $I_{1j}$ for $j=1,2,\dots,d$ is independent across $j$ and is given by:
    \begin{align*}
        P[I_{1j} = 1 \mid \beta,\Omega_1,Y] &= \dfrac{q_n \phi_N(\beta_j,0,\tau_{1n}^2)}{(1-q_n) \phi_N(\beta_j,0,\tau_{0n}^2) + q_n \phi_N(\beta_j,0,\tau_{1n}^2)}.
    \end{align*}
    
\item \textbf{Sampler for the conditional of $\gamma$}:\\

Sampling of $\gamma$ \eqref{gamma-posterior} from the Model \eqref{gamma_working} is based on a logistic posterior distribution. We use the P\'{o}lya-Gamma latent variables (denote by $\omega_{21}, \dots \omega_{2n}$) to accomplish this task. In the prior specification of $\gamma$ \eqref{gamma-prior}, the binary latent variables for spike and slab are denoted by $I_2$.
For $i = 1, 2, \dots, n$, the conditional distribution of P\'{o}lya-Gamma random variables is given by:
$$\omega_{2i} \mid (\tilde{\eta}, Z_i,I_2) \sim \text{PG}(1, Z_{i}^T\gamma).$$
Define $\Omega_2 := \text{Diag}(\omega_{21}, \dots \omega_{2n}) $.
The conditional distribution for $\gamma$ is given by:   
    \begin{equation*}
        \gamma \mid (I_2,X,Z,\Omega_2) \sim N\bigg((Z^T\Omega_2Z + D_{I_2})^{-1}(Z^T(X-1/2)),(Z^T\Omega_2Z + D_{I_2})^{-1}\bigg),
    \end{equation*}
where $D_{I_2} = Diag(I_2 \tau_{1n}^{-2} + (1-I_2)\tau_{0n}^{-2})$.

Let $\phi_N(r,\mu,\sigma^2)$ denotes the standard Gaussian density function of variable $r$ with mean $\mu$ and variance $\sigma^2$. The conditional distribution for indicators $I_{2j}$ for $j=1,2,\dots,d$ is independent across $j$ and is given by:
    \begin{align*}
        P[I_{2j} = 1 \mid h(Z),\Omega_2,X,Z] &= \dfrac{q_n \phi_N(\gamma_j,0,\tau_{1n}^2)}{(1-q_n) \phi_N(\gamma_j,0,\tau_{0n}^2) + q_n \phi_N(\gamma_j,0,\tau_{1n}^2)}.
    \end{align*}

\item  \textbf{Sampler for the conditional of $\theta$:} \\
    The sampling of $\theta$ from the model \eqref{theta_posterior} is based on a logistic posterior distribution too. We introduce another set of P\'{o}lya-Gamma latent variables $\omega_3 := (\omega_{31}, \dots, \omega_{3n})$. We use $\phi_i := \tilde{\theta}h(Z_i) + Z_i^T\tilde{\beta}$ when defining the posterior for $\theta$. For $i = 1, 2, \dots, n$, the conditional distribution of P\'{o}lya-Gamma latent variables is given by:
    \begin{equation*}
        \omega_{3i}\mid (\theta,Y,X,Z,h(Z),\phi) \sim \text{PG}(1, \theta (X_i - h(Z_i)) + \phi_i ).
    \end{equation*}
Let $\Omega_3 = \text{Diag}(\omega_{31}, \dots, \omega_{3n})$. The conditional posterior for $\theta$ is given by:
\begin{equation}\label{q_posterior}
    \theta \mid (\Omega_3, Y,X,Z,h(Z), \phi) )\sim N \bigg( \frac{\Tilde{z}^T\Omega_3\Tilde{X}}{\Tilde{X}^T\Omega_3\Tilde{X} + \lambda^{-1}}, \frac{1}{\Tilde{X}^T\Omega_2\Tilde{X} + \lambda^{-1}} \bigg),
\end{equation}
where $ \Tilde{z}_i = (y_i - 1/2)/\omega_{3i} - \phi_i$ and $\tilde{X}_i = X_i - h(Z_i)$ for $i = 1, \dots, n$.
\end{enumerate}

\section{Technical Lemmas and Proofs}\label{proofs}
\subsection{Additional Notations and Definitions}

We present some more notations that we need in the proofs. We define 
\begin{equation}
    \mu_i(\theta, h(Z), \phi) = \exp\{\theta (X_i - h(Z_i)) + \phi_i\}/[1 + \exp\{\theta (X_i - h(Z_i)) + \phi_i\}],
\end{equation}
and $\Sigma(\theta, h(Z), \phi)$ be a diagonal matrix with diagonal entries being
\begin{equation}
    \{\Sigma(\theta, h(Z), \phi)\}_{i,i} = \mu_i(\theta, h(Z), \phi)(1-\mu_i(\theta, h(Z), \phi)).
\end{equation}
For two sequences $a_n$ and $b_n$, we use the notation $a_n = (1 \pm \epsilon)b_n$ as a shorthand for the inequality: $(1-\epsilon)b_n \leq a_n \leq (1+\epsilon)b_n$.

\subsection{Auxiliary Lemmas}\label{lemmasec}
\begin{lemma}
\label{lemma_walker}
For every $\delta >0 $, there exists $k_1(\delta)>0$ such that
\begin{align*}
            &\lim_{n\to \infty} P \bigg(   \sup_{\theta \in \Theta - N_0(\delta)} n^{-1} (\text{L}_n(\theta \mid h_0(Z), \phi_0) - \text{L}_n(\theta_0 \mid h_0(Z), \phi_0)) < - k_1(\delta) \bigg) = 1,\\
            &n^{-1} \text{L}_n^{''}(\mle \mid h_0(Z), \phi_0) \prob -I(\theta_0) \text{, and}\\
            &\text{L}_n(\theta_0 \mid h_0(Z), \phi_0) - \text{L}_n(\mle \mid h_0(Z), \phi_0) = O_P(1).
\end{align*}
\end{lemma}
Lemma ~\ref{lemma_walker} consists of preliminary results shown in \citep{walker,schervish2012theory}.

\begin{lemma}\label{lemma_5_ldml} Consider the notations $\mathbb{E}_n(r_i) = \frac{1}{n}\sum_{i=1}^{n}r_i$ and $\bar{E}(r_i) =\frac{1}{n}\sum_{i=1}^{n}E[r_i] $.
Let $|h_i(t)| \leq |t^TW_i|, \ K/4 > \bar{\sigma}^2 = \sup_{t \in \mathcal{T}} \bar{E}[h_i(t)^2\xi_i^2],\ \tilde{p} = dim(W_i),$ and $\|\mathcal{T}\|_1 = \sup_{t \in \mathcal{T}}\|t\|_1.$ We have

\begin{align*}
    P\bigg( \sup_{t \in \mathcal{T}} |(\mathbb{E}_n - \bar{E})[h_i(t)\xi_i]| > \dfrac{K\|\mathcal{T}\|_1\sqrt{M}}{\sqrt{n}}  \bigg) &\leq 32 \tilde{p}\exp\bigg(\dfrac{-K^2}{16} \bigg) +\\
    & P\bigg(\max_{j \leq \tilde{p}}\mathbb{E}_n[W_{ij}^2\xi_i^2] >M \bigg).
\end{align*} 
\end{lemma}

Lemma~\ref{lemma_5_ldml} has been proved in  \citep{logistic_dml}. 

\begin{lemma}\label{h_phi_conc}
    Let the regularity assumptions~\ref{ass1}-\ref{ass5} hold. Then, the following concentration result holds for the estimated quantities $h(Z)$ and $\phi$ around their true values $h_0(Z)$ and $\phi_0$, respectively:

\begin{align}
        P\bigg( \max_{i \leq n}|h(Z_i) - h_0(Z_i)|  &\geq Ms \sqrt{\frac{\log(d \vee n)}{n}} \mid Y,X,Z \bigg) 
        \leq C (d \vee n)^{-c_3},
\end{align}
    on a set $\mathcal{E}_3$, with $P(\mathcal{E}_3) \geq 1 - (d \vee n)^{-c_3}$, and 
\begin{equation}
        P\bigg( \max_{i \leq n}|\phi_i - \phi_{0i}|  \geq Ms \sqrt{\frac{\log(d \vee n)}{n}} \mid Y,X,Z \bigg)  \leq C (d \vee n)^{-c_4},
\end{equation}
    on a set $\mathcal{E}_4$, with $P(\mathcal{E}_4) \geq 1 - (d \vee n)^{-c_4}$, where $M,C,c_3$ and $c_4$ are some positive constants.
    
\end{lemma}

\begin{proof}[Proof of Lemma~\ref{h_phi_conc}]

Recall the definitions of $h_0(Z_i)$ and $h(Z_i)$ from \eqref{h_fn} and \eqref{h_estimate}, respectively:
$$h(Z_i) = \frac{1}{1+R_i}, \ \ h_0(Z_i) = \frac{1}{1+R_{0i}},$$
where $$R_{0i} = \dfrac{P(Y_i=1\mid X_i=0, Z_i)P(Y_i=0\mid X_i=0, Z_i)P(X_i=0\mid Z_i)}{P(Y_i=1\mid X_i=1, Z_i)P(Y_i=0\mid X_i=1, Z_i)P(X_i=1\mid Z_i)},$$ and 
$$R_i = \dfrac{P(Y_i=1\mid X_i=0, Z_i, \tilde{\theta},\beta)P(Y_i=0\mid X_i=0, Z_i, \tilde{\theta},\beta)P(X_i=0\mid Z_i, \gamma)}{P(Y_i=1\mid X_i=1, Z_i,\tilde{\theta},\beta)P(Y_i=0\mid X_i=1, Z_i,\tilde{\theta},\beta)P(X_i=1\mid Z_i,\gamma)}.$$
For positive real numbers $x$ and $y$, we can easily show that $$\Bigg|\frac{1}{1+x} - \frac{1}{1+y}\Bigg| \leq \Bigg| \frac{x}{y} - 1 \Bigg|.$$
Therefore,
$$|h(Z_i) - h_0(Z_i)| \leq \Bigg| \frac{R_{0i}}{R_i} - 1 \Bigg|.$$
To obtain the desired result, it is enough to show that for all $i = 1,2,\dots,n$:
$$ \Bigg|\log\Bigg( \frac{R_{0i}}{R_i} \Bigg)\Bigg| = O_P\Bigg( s\sqrt{\frac{\log (d \vee n}{n}}\Bigg).$$
We can split $\log(R_{0i}/R_i)$ into three parts as follows:
\begin{align*}
    \log\Bigg( \frac{R_{0i}}{R_i} \Bigg) &= \underbrace{\log \Bigg( \frac{P(Y_i=1\mid X_i=0, Z_i)P(Y_i=0\mid X_i=0, Z_i)}{P(Y_i=1\mid X_i=0, Z_i, \tilde{\theta},\beta)P(Y_i=0\mid X_i=0, Z_i, \tilde{\theta},\beta)} \Bigg)}_{T_1} \\
    &+ \underbrace{\log \Bigg( \frac{P(Y_i=1\mid X_i=1, Z_i,\tilde{\theta},\beta)P(Y_i=0\mid X_i=1, Z_i,\tilde{\theta},\beta)}{P(Y_i=1\mid X_i=1, Z_i)P(Y_i=0\mid X_i=1, Z_i)} \Bigg)}_{T_2} \\
    &+ \underbrace{\log \Bigg( \frac{P(X_i=1\mid Z_i,\gamma)/P(X_i=0\mid Z_i,\gamma)}{P(X_i=1\mid Z_i)/P(X_i=0\mid Z_i)}\Bigg)}_{T_3}
\end{align*}
Based on the logistic expressions of the probabilities and using the inequality, $(1+e^a)/(1+e^b) \leq e^{|a-b|}$, we can show that:
\begin{align*}
    |T_1| &\leq 3 |Z_i^T(\beta - \beta_0)| \leq K \|\beta - \beta_0\|_1 (K>0),\\
    &= O_P\Bigg( s\sqrt{\frac{\log (d \vee n}{n}}\Bigg), 
    % (\text{Assumption).
\end{align*}
where the last equality is due to Assumption~\ref{ass4}.
Similarly, we can show that:
\begin{align*}
    |T_2| &\leq 3 | (\tilde{\theta} - \theta_0)X_i + Z_i^T(\beta - \beta_0)| \leq K \|(\tilde{\theta},\beta) - (\theta_0,\beta_0)\|_1 (K>0),\\
    &= O_P\Bigg( s\sqrt{\frac{\log (d \vee n}{n}}\Bigg),
\end{align*}
where the last equality is due to Assumption~\ref{ass4}. Finally, Assumption~\ref{ass_propensity} to obtain the same concentration rate for $T_3$.  

For the concentration rate of $\phi$, note that:
\begin{align*}
    |\phi_i - \phi_{0i}| &= |\tilde{\theta}h(Z_i) + Z_i^T\beta - \theta_0h_0(Z_i) - Z_i^T\beta_0| \\
    &\leq |\tilde{\theta} - \theta_0||h(Z_i)| + |\theta_0||h(Z_i) - h_0(Z_i)| + |Z_i^T(\beta - \beta_0)| \\
    &=  O_P\Bigg( s\sqrt{\frac{\log (d \vee n}{n}}\Bigg),
\end{align*}
where for the last equality, we have used the fact that $|h(Z_i)| < 1$, Assumption~\ref{ass4}, Assumption~\ref{ass5}, and the concentration rate of $h(Z_i)$ from the previous part.
\end{proof}
\begin{lemma}\label{lemma_hessians}
Suppose the nuisance parameters $\gamma$ and $\phi$ satisfy the Assumptions~\ref{ass1}-\ref{ass4} and $h_0(Z),\phi_0$ are the corresponding oracle parameters.
Let $$M_n(\theta,h(Z), \phi) = Z^T\Sigma(\theta,h(Z), \phi) Z,\ \  M_n(\theta,h_0(Z),\phi_0) = Z^T\Sigma(\theta,h_0(Z),\phi_0)Z.$$
For every $\theta \in \Theta$, there exists $\epsilon_n(\theta) \to 0$ such that
\begin{equation}
    (1-\epsilon_n(\theta))M_n(\theta,h_0(Z),\phi_0) \leq M_n(\theta,h(Z), \phi) \leq (1+\epsilon_n(\theta))M_n(\theta,h_0(Z),\phi_0).
\end{equation}
Furthermore, for every $\epsilon>0$, there exists $\delta >0$ such that when $\theta$ satisfies $|\theta - \theta_0| < \delta$,
\begin{equation}
    (1-\epsilon)M_n(\theta,h_0(Z),\phi_0) \leq M_n(\theta,h(Z), \phi) \leq (1+\epsilon)M_n(\theta,h_0(Z),\phi_0).
\end{equation}
\end{lemma}

\begin{proof}[Proof of Lemma~\ref{lemma_hessians}]
To prove the first statement, it is sufficient to show that 
\[ (1-\epsilon_n(\theta))\{\Sigma(\theta,h_0(Z),\phi_0)\}_{i,i} \leq \{\Sigma(\theta,h(Z), \phi)\}_{i,i} \leq (1+\epsilon_n(\theta))\{\Sigma(\theta,h_0(Z),\phi_0)\}_{i,i},\]
for each $i = 1, \dots, n.$ Using the fact that $(1+e^a)/(1+e^b) \leq e^{|a-b|}$, we have

\begin{align*}
    \dfrac{\{\Sigma(\theta,h(Z), \phi)\}_{i,i}}{\{ \Sigma(\theta,h_0(Z),\phi_0)\}_{ii}} &\leq \exp\{ 3 (|\theta| |h_0(Z_i) - h(Z_i)| + |Z_i^T(\phi - \phi_0|)\} \\ \leq &\exp \{C(|\theta|\|h_0(Z_i) - h(Z_i)\|_1 + \|\phi - \phi_0\|_1 \} \leq \exp\{ C (1+|\theta|)\sqrt{\log(d \vee n)/n}\},
\end{align*} 
which converges to 1 for each $\theta$ since $\gamma$ and $\phi$ follow their regularity conditions and  $|Z_{ij}| \leq C$. Hence for some $C>0$ which is independent of $\theta$, $\epsilon_n(\theta) =  C (1+|\theta|)(\log(d \vee n)/n)^{1/2} $ works. By interchanging the terms in the ratio we would obtain the reverse inequality.
\newline

For the second statement, we follow the same steps as above to get:
\begin{align*}
    \dfrac{\{\Sigma(\theta,h(Z), \phi)\}_{i,i}}{\{ \Sigma(\theta,h_0(Z),\phi_0)\}_{ii}} &\leq \exp\{ C (1+|\theta|)\sqrt{\log(d \vee n)/n}\}\\
    &\leq \exp\{ C (1+|\theta_0| + |\theta - \theta_0|)\sqrt{\log(d \vee n)/n}\} \\
    &\leq\exp\{ C (1+|\theta_0| + \delta)\sqrt{\log(d \vee n)/n}\}.
\end{align*}

Therefore, $\delta$ can be chosen small enough to make $C (1+|\theta_0| + \delta)(\log(d \vee n)/n)^{1/2}$ less than $\epsilon$. By interchanging the terms in the ratio we would obtain the reverse inequality.
\end{proof}

\begin{lemma}\label{lemma_scores}
As in Lemma~\ref{lemma_hessians}, consider $M_n(\theta,h_0(Z),\phi_0) = Z^T\Sigma(\theta,h_0(Z),\phi_0)Z.$
Then for every $\epsilon>0$, there exists $\delta >0$ such that when $|\theta - \theta_0| < \delta$, we have
\[ (1- \epsilon)M_n(\theta_0,h_0(Z),\phi_0) \leq M_n(\theta,h_0(Z),\phi_0) \leq (1+\epsilon)M_n(\theta_0,h_0(Z),\phi_0).\] 

\end{lemma}

\begin{proof}[of Lemma~\ref{lemma_scores}]

It is sufficient to show that a $\delta>0$ can be chosen such that for $\{\theta :|\theta -\theta_0| < \delta\}$, $$\dfrac{\{\Sigma(\theta,h_0(Z),\phi_0)\}_{i,i}}{\{ \Sigma(\theta_0,h_0(Z),\phi_0)\}_{ii}} \to 1.$$
We again use the fact that $(1+e^a)/(1+e^b) \leq e^{|a-b|}$.
From the definition of $\{\Sigma(\theta,h_0(Z),\phi_0)\}_{i,i}$, for $K>0$
\[ \dfrac{\{\Sigma(\theta,h_0(Z),\phi_0)\}_{i,i}}{\{ \Sigma(\theta_0,h_0(Z),\phi_0)\}_{ii}} \leq \exp\{3|\theta_0 - \theta||X_i - h_0(Z_i)|\} \leq \exp\{K\delta\}.\]
The quantity $\delta$ can be chosen appropriately so that the difference of this ratio and 1 is less than $\epsilon$. To be more specific, if $0 < \epsilon <1$, any $\delta$ such that $\log(1-\epsilon)/K < \delta < \log(1+\epsilon)/K$ would suffice. By interchanging the terms in the ratio we would obtain the reverse inequality. 
\end{proof}

% \begin{lemma}\label{lemma_ldiff_theta_j}
% Let $\gamma$ and $\phi$ follow the Assumptions~\ref{ass1}-\ref{ass4}. For any $\theta_j \in \Theta$ and $\delta >0$, let $N_j(\delta) = \{ \theta : |\theta - \theta_j| < \delta\},$ then for some $K>0$
% \begin{equation}
%     \sup_{\theta \in N_j(\delta)} |n^{-1}(\text{L}_n(\theta \mid h(Z), \phi) - \text{L}_n(\theta \mid h_0(Z), \phi_0) | \leq Ks(1+\delta + |\theta_j|)\sqrt{\dfrac{\log(d \vee n)}{n}}.
% \end{equation}
% \end{lemma}

% \begin{proof}[Proof of Lemma~\ref{lemma_ldiff_theta_j}]
% Take $\theta \in N_j(\delta)$. Using first order Taylor's expansion of the log likelihood around the true nuisance parameter values and keeping $\theta$ fixed, we get
% \begin{align*}
%     n^{-1}|{\text{L}_n(\theta \mid h(Z), \phi)}_{} - {L_{n}(\theta\midh_0(Z),\phi_0)}_{}| &\leq n^{-1}|(\gamma - \gamma_0)^Ts_{\gamma}(\theta, \tilde{\gamma}, \tilde{\phi})| + 
% n^{-1}|(\phi - \phi)^Ts_{\phi}(\theta, \tilde{\gamma}, \tilde{\phi})| \\
% &\leq \|\gamma - \gamma_0\|_1 (n^{-1}\|s_{\gamma}(\theta, \tilde{\gamma})\|_{\infty}) \\
% &+ \|\phi - \phi_0\|_1(n^{-1}\|s_{\phi}(\theta, \tilde{\gamma})\|_{\infty}) \\
% &\leq Ks(|\theta|+1)\sqrt{\frac{\log(d \vee n)}{n}} \\
% &\leq Ks(1+\delta + |\theta_j|)\sqrt{\frac{\log(d \vee n)}{n}}.
% \end{align*}
% Here we have used the fact that $Z^T(Y-\mu) = O(n)$.
% \end{proof}

\begin{lemma}[Difference of log likelihoods]\label{lemma_ldiff_complement}
Let $\gamma$ and $\phi$ satisfy the regularity assumptions \ref{ass1}-\ref{ass5}. Then, for all $\delta >0,$ there exists $k(\delta) >0 $ such that
\begin{equation}
    \lim_{n\to \infty} P \bigg(   \sup_{\theta \in \Theta - N_0(\delta)} n^{-1} (\text{L}_n(\theta \mid h(Z), \phi) - \text{L}_n(\theta_0 \mid h_0(Z), \phi_0)) < - k(\delta) \bigg) = 1.
\end{equation}
\end{lemma}

\begin{proof}[Proof of Lemma~\ref{lemma_ldiff_complement}]
For any $\theta \in \Theta-N_0(\delta),$ define 
$$R_i = \log\bigg( \dfrac{p_n(Y_i \mid \theta,h(Z), \phi,X_i,Z_i)}{p_n(Y_i \mid \theta_0, h_0(Z), \phi_0,X_i,Z_i)} \bigg).$$ 
% \newline
% \textbf{Case 1: If $E[Z_i]$ is finite:}
% \newline
By strict concavity of the function $\log$, 
\begin{align*}
    E[R_i] &< \log[E\{\exp(R_i)\}] \\
    &= \log \bigg( \sum_{Y_i \in \{0,1\}} \dfrac{p_n(Y_i \mid \theta,h(Z), \phi,X_i,Z_i)}{p_n(Y_i \mid \theta_0, h_0(Z), \phi_0,X_i,Z_i)} p_n(Y_i \mid \theta_0, h_0(Z), \phi_0,X_i,Z_i) \bigg)\\
    &= 0.
\end{align*}
Applying weak law of large numbers to $n^{-1}\sum R_i,$ we can find $k(\delta) >0$ such that 
\begin{equation}\label{l_diff_finite}
    \lim_{n\to \infty} P \bigg(  n^{-1} (\text{L}_n(\theta \mid h(Z), \phi) - \text{L}_n(\theta_0 \mid h_0(Z), \phi_0)) < - k(\delta)< 0\bigg) = 1.
\end{equation}
$E[R_i]$ is always finite since $E[R_i]$ is Kullback–Leibler divergence between $p_n(Y_i \mid \theta_0, h_0(Z), \phi_0)$ and  $p_n(Y_i \mid \theta,h(Z), \phi)$ and based on the expression of logistic probabilities $0< p_n(Y_i \mid \theta_0, h_0(Z), \phi_0) <1$ and $0 < p_n(Y_i \mid \theta,h(Z), \phi) <1$.

\end{proof}

\begin{lemma}[Behavior of expectation of score differences]\label{lemma_score_diff_expectation}
Let $h(Z)$ and $\phi$ follow the Assumptions~\ref{ass1}-\ref{ass5}. For given $\{Y_1,\dots, Y_n\}$,$\{Z_1, \dots, Z_n\}$ and $\{X_1, \dots, X_n\}$ consider
$$b_n(h(Z), \phi) = \frac{1}{n}\sum_{i=1}^{n}(X_i-h(Z_i))(Y_i-\mu_i(\theta_0,h(Z_i), \phi_i)),$$ 
$$\text{ where \ \ \ } \mu_i(\theta_0,h(Z_i), \phi_i) = \frac{\exp(\theta_0 (X_i - h(Z_i)) + \phi_i)}{1+\exp(\theta_0 (X_i - h(Z_i)) + \phi_i)}.$$ Then, 
\begin{equation}\label{mean_behavior}
    \sqrt{n}|E[b_n(h(Z), \phi) - b_n(h_0(Z),\phi_0)]| \to 0 \text{ as } n \to \infty, \text{ and }
\end{equation}
\begin{equation}\label{faster_clt}
    \sqrt{n}|[b_n(h(Z), \phi) - b_n(h_0(Z),\phi_0)] - E[b_n(h(Z), \phi) - b_n(h_0(Z),\phi_0)]| \to 0 \text{ as } n \to \infty
.\end{equation}
\end{lemma}

\begin{proof}[Proof of Lemma \ref{lemma_score_diff_expectation}]
A version of the results in this lemma have been proved in \citep{logistic_dml}. We will use second order Taylor's series expansion with directional derivatives. The directional derivative of $f$ at $\tilde{x}$ in the direction of $[x - x_0]$ is defined as:
$$\partial_tf(\tilde{x})[x - x_0] = \lim_{t \to 0}\frac{f(\tilde{x} + t[x - x_0]) -f(\tilde{x})}{t}.$$
Similarly, the second directional derivative can be defined as:
$$ \partial^2_tf(\tilde{x})[x - x_0,x - x_0] = \lim_{t \to 0}\frac{\partial_tf(\tilde{x} + t[x - x_0])- \partial_tf(\tilde{x})}{t}.$$
These two derivatives are functions of $\tilde{x}$ and $[x- x_0]$ is used in the notation to specify the direction. For simplicity of notations, denote $\eta = (\eta_1,\dots,\eta_n)^T$ and $\eta_0 = (\eta_{01},\dots,\eta_{0n})^T$, where $\eta_i = (h(Z_i),\phi_i)$ and $\eta_{0i} = (h_0(Z_i),\phi_{0i})$ for $i=1,\dots,n$. Now, using Taylor's series (around $\eta_{0i}$ for $i=1,\dots, n$) we can write 
\[ b_n(\eta) - b_n(\eta_0) = \partial_tb_n(\eta_0)[\eta - \eta_0] + \partial^2_tb_n(\tilde{\eta})[\eta-\eta_0,\eta-\eta_0],\]
for some $\tilde{\eta} \in [\eta_0,\eta]$, where $\tilde{\eta_i} = (\tilde{h}(Z_i),\tilde{\phi_i})$. For $i=1,\dots,n$, we can calculate the directional derivatives of $f(\tilde{\eta}_i) = (X_i-\tilde{h}(Z_i))(Y_i-\mu_i(\theta_0,\tilde{h}(Z_i), \tilde{\phi_i}))$ at $\tilde{\eta}_i$ in the direction of $[\eta_i - \eta_{0i}]$ and obtain:
\begin{align*}
    \partial_tb_n(\tilde{\eta})[\eta - \eta_0] = -\dfrac{1}{n}\sum_{i=1}^{n}&\bigg[\Big((\phi_i - \phi_{0i})- \theta_0(h(Z_i) - h_0(Z_i))\Big)(X_i - \tilde{h}(Z_i))\mu_i^{\prime}(\theta_0,\tilde{\eta})\\
    &+ (h(Z_i) - h_0(Z_i))(Y_i - \mu_i(\theta_0,\tilde{\eta}) \bigg] \text{ and}
\end{align*}
\begin{align*}
    \partial^2_tb_n(\tilde{\eta})[\eta-\eta_0,\eta-\eta_0] &= \dfrac{1}{n}\sum_{i=1}^{n}\bigg[(X_i -\tilde{h}(Z_i))\Big((\phi_i - \phi_{0i}) - \theta_0 (h(Z_i) - h_0(Z_i))\Big)^2 \mu_i^{\prime\prime}(\theta_0,\tilde{\eta}) \\
    &+ 2\Big((\phi_i - \phi_{0i}) - \theta_0 (h(Z_i) - h_0(Z_i))\Big)((h(Z_i) - h_0(Z_i)))\mu_i^{\prime}(\theta_0,\tilde{\eta})\bigg],
\end{align*} 
where $\mu^{\prime} = \mu(1-\mu)$ and $\mu^{\prime\prime} = \mu(1-\mu)(1 - 2\mu)$ which are both bounded above by 1. By Neyman orthogonality, we have $E[\partial_tb_n(\eta_0)[\eta - \eta_0]] = 0.$ Moreover, since $\tilde{\eta} \in [\eta_0,\eta], |X_i - \tilde{h}(Z_i)| \leq |X_i - h_0(Z_i)| + |(h(Z_i) - h_0(Z_i))|,$ Jensen's inequality and boundedness of $\mu^{\prime}$ and $\mu^{\prime\prime}$ imply
\begin{align*}
    &|E[b_n(\eta) - b_n(\eta_0)]| \leq \sup_{\tilde{\eta} \in [\eta_0,\eta]}E\Big[\big|\partial^2_tb_n(\tilde{\eta})[\eta-\eta_0,\eta-\eta_0]\big|\Big] \\
    &\leq \dfrac{1}{n}\sum_{i=1}^{n}E\bigg[|X_i - h_0(Z_i)| \Big\{ (\phi_i - \phi_{0i}) - \theta_0 (h(Z_i) - h_0(Z_i))\Big\}^2  \bigg] \\
    &+ \dfrac{1}{n}\sum_{i=1}^{n}E\bigg[ |(h(Z_i) - h_0(Z_i))|\Big\{ (\phi_i - \phi_{0i}) - \theta_0 (h(Z_i) - h_0(Z_i))\Big\}^2\bigg] \\ 
    &+  E\bigg[\dfrac{2}{n}\sum_{i=1}^{n} |(h(Z_i) - h_0(Z_i))|\Big| (\phi_i - \phi_{0i}) - \theta_0 (h(Z_i) - h_0(Z_i))\Big|\bigg] \\
    &\leq \max_{i\leq n}E[|X_i - h_0(Z_i)|]E\bigg[\dfrac{1}{n}\sum_{i=1}^{n}\Big\{ (\phi_i - \phi_{0i}) - \theta_0 (h(Z_i) - h_0(Z_i))\Big\}^2  \bigg] \\
    &+ \max_{i\leq n}E[|(h(Z_i) - h_0(Z_i))|]E\bigg[\dfrac{1}{n}\sum_{i=1}^{n} \Big\{ (\phi_i - \phi_{0i}) - \theta_0 (h(Z_i) - h_0(Z_i))\Big\}^2\bigg]\\
    &+ 2E\bigg[\bigg(\dfrac{1}{n}\sum_{i=1}^{n} \{(h(Z_i) - h_0(Z_i))\}^2\bigg)^{1/2}\bigg(\dfrac{1}{n}\sum_{i=1}^{n}\Big\{ (\phi_i - \phi_{0i}) - \theta_0 (h(Z_i) - h_0(Z_i))\Big\}^2\bigg)^{1/2}\bigg] \\
    &\leq_p K(1+|\theta_0|)\frac{s\log(d \vee n)}{n}, \text{ for some } K>0.
\end{align*}

Here, we have used the following inequalities based on the regularity Assumption~\ref{ass4} and Lemma~\ref{h_phi_conc}:
\begin{align*} 
&\dfrac{1}{n}\sum_{i=1}^{n}\|\eta_i - \eta_{0i}\|_2^2 = O_p\Big(\frac{s^2\log(d \vee n)}{n}\Big),\\  &|h(Z_i)- h_0(Z_i)| = O_p(s\sqrt{\log(d \vee n)/n}) \text{, and } E[|X_i - h_0(Z_i)|] \leq C.
\end{align*}
Therefore, due to Assumption~\ref{ass3}
\[\sqrt{n}|E[b_n(h(Z), \phi) - b_n(h_0(Z),\phi_0)]| \leq K\dfrac{s^2\log(d \vee n)}{\sqrt{n}} \to 0,\]
which proves \eqref{mean_behavior}.
 To prove \eqref{faster_clt}, which is a result faster than CLT, we will use Lemma~\ref{lemma_5_ldml}. First, we will rewrite the difference between $b_n(h(Z), \phi)$ and $b_n(h_0(Z),\phi_0)$ in another way. For simplicity, recall the notations $\mathbb{E}_n(r_i) = \frac{1}{n}\sum_{i=1}^{n}r_i$ and $\bar{E}(r_i) =\frac{1}{n}\sum_{i=1}^{n}E[r_i] $. Therefore,
\begin{align}
    &|[b_n(h(Z), \phi) - b_n(h_0(Z),\phi_0)] - E[b_n(h(Z), \phi) - b_n(h_0(Z),\phi_0)]| \nonumber \\
    &= \big|(\mathbb{E}_n - \bar{E})\Big((X_i - h(Z_i))(Y_i - \mu_i) - (X_i - h_0(Z_i))(Y_i - \mu_{0i})\Big)\big|\nonumber \\
    &= |(\mathbb{E}_n - \bar{E})\Big( (\mu_{0i}- \mu_i)(h_0(Z_i) - h(Z_i))\Big)|\\
    &+ |(\mathbb{E}_n - \bar{E})\Big((Y_i - \mu_{0i})(h_0(Z_i) - h(Z_i))) \Big)|\\
    &+ |(\mathbb{E}_n - \bar{E})\Big((\mu_{0i} - \mu_i)(X_i - h_0(Z_i))) \Big)|.
\end{align}
Due to Lipschitz property of $\mu$, the first term becoms similar to the one in the first part of this lemma. The second and third terms will be controlled based on Lemma~\ref{lemma_5_ldml}. 
\end{proof}

\subsection{Proof of Theorem}\label{proofsec}

\begin{proof}[Proof of Theorem 1]
Consider $\hat{\theta}_0$ be the MLE based on the conditional  likelihood when the true nuisance parameters $\gamma_0$ and $\phi_0$ are known. 
Take $\mathcal{E} = \mathcal{E}_1 \cap \mathcal{E}_2$ where $\mathcal{E}_1$ and $\mathcal{E}_2$ are sets from Assumption~\ref{ass4}. We can show that $P(\mathcal{E}) \geq 1 - (d \vee n)^{-c_3}$, where $c_3 = \max \{c_1,c_2\}$. We proceed with proving the theorem on the set $\mathcal{E}$ to get the desired result.
On set $\mathcal{E}$, consider the following set which corresponds to the cases when nuisance parameters are around their oracle values:
$$B := \bigg\{ (\tilde{\theta},\beta,\gamma) : \|h(Z) - h_0(Z)\|_1 \leq Ms_1\sqrt{\frac{\log (d \vee n)}{n}}, \|\phi - \phi_0\|_1 \leq Ms_2 \sqrt{\frac{\log (d \vee n)}{n}}\bigg\}.$$
Based on Assumption~\ref{ass4} and property of $\mathcal{E}$, $\Pi(B^c\mid X,Y,Z) \leq (d \vee n)^{-c_3}.$ Moreover, for any measurable set $A$ in $\mathbb{R}$, 
\begin{align}\label{prob_diff}
  & |\Pi(\theta \in A\mid X,Y,Z) - \Pi(\theta \in A\mid X,Y,Z,B)| \nonumber\\
  &= |\Pi(\theta \in A\mid X,Y,Z,B)\Pi(B\mid X,Y,Z) \nonumber\\ &+\Pi(\theta \in A\mid X,Y,Z,B^c)\Pi(B^c\mid X,Y,Z)  - \Pi(\theta \in A\mid X,Y,Z,B)| \nonumber\\
  &\leq 2\Pi(B^c\mid X,Y,Z) \leq  (d \vee n)^{-c_3}.
\end{align}
This shows that the measure of any set given marginal posterior of $\theta$ is close to the measure of that set under conditional posterior of $\theta$ when $\gamma$ and $\phi$ are from high probability set $B$. In the rest of the proof, we establish normality of the conditional posterior for the nuisance parameters in set $B$.

Following the notation of log likelihood from \eqref{log-l-notation}, we can rewrite the likelihood in the conditional posterior as:
\begin{align}
    &\underbrace{p_n(Y\mid\theta, h(Z), \phi)}_{\exp(\lest)} =  \underbrace{p_n(Y\mid\hat{\theta}_0,h_0(Z),\phi_0)}_{\exp(\lmle)} \exp(\lest - \lmle) \label{mle_split}\\
    &= p_n(Y\mid\hat{\theta}_0,h_0(Z),\phi_0) \exp(\lest - \ltrue) \times \nonumber
    \exp(\ltrue - \lmle) \nonumber\\
    &= p_n(Y\mid\hat{\theta}_0,h_0(Z),\phi_0)\exp(\lest - \ltrue)\exp\bigg[\dfrac{-(\theta - \mle)^2(1+R_n)}{2\sigma_n^2} \bigg]\label{taylor_split},
\end{align}
where 
\[R_n = R_n(\theta,h_0(Z),\phi_0) = \sigma_n^2(\text{L}_n^{''}(\theta^*_n \mid h_0(Z),\phi_0) - \text{L}_n^{''}(\mle \mid h_0(Z),\phi_0)),\] $\text{ with  } \theta^*_n \in (\theta,\mle) \text{ and } 
(\sigma_n^2)^{-1} = \text{L}_n^{''}(\mle \mid h_0(Z), \phi_0)$.
Here the derivatives are with respect to the parameter of interest $\theta$. We can split the marginal into two parts as:
\begin{align}\label{I_split}
    \int_{\Theta}\pi(\theta)\exp(\text{L}_n(Y\mid X,Z,\theta, h(Z), \phi) d\theta &= \nonumber\\
    \underbrace{\int_{\Theta - N_0(\delta)}\pi(\theta)\exp(\lest) d\theta}_{I_1}
    &+\underbrace{\int_{N_0(\delta)}\pi(\theta)\exp(\lest) d\theta}_{I_2}.
\end{align}
The choice of $\delta$ will be specified towards the end. Whenever needed, we will keep updating the value of $\delta$. In the end, it will be chosen as minimum of a few $\delta$ values.

\textbf{Behavior of $I_2$:} For $I_2$, we will use the expression in \eqref{taylor_split} and rewrite the integral as:
\begin{align*}
    I_2 =  p_n(Y\mid\hat{\theta}_0,h_0(Z),\phi_0)\int_{N_0(\delta)}&\pi(\theta)\exp(\lest - \ltrue) \times\\
    &\exp\bigg[\dfrac{-(\theta - \mle)^2(1+R_n)}{2\sigma_n^2} \bigg] d\theta.
\end{align*}
By the continuity of the prior $\pi(\theta)$ at $\theta_0$ in Assumption~\ref{ass5}, we have for every $\epsilon >0$, there exists $\delta >0$ such that

\begin{equation}
    |\pi(\theta) - \pi(\theta_0)| < \epsilon \pi(\theta_0) \text{, if } \theta \in N_0(\delta)
.\end{equation}
Then, we have
\begin{equation}\label{I_2_to_I_3}
    (1-\epsilon)I_3 < \{\pi(\theta_0) p_n(Y\mid\hat{\theta}_0,h_0(Z),\phi_0) \}^{-1}I_2 < (1+\epsilon)I_3, \text{ where }
\end{equation}

\begin{equation}\label{I_3}
    I_3 = \int_{N_0(\delta)}\exp(\lest - \ltrue)\exp\bigg[\dfrac{-(\theta - \mle)^2(1+R_n)}{2\sigma_n^2} \bigg] d\theta
.\end{equation}

It can be shown that \citep{walker} or every $\epsilon >0$, there exists $\delta >0$ such that
\begin{equation}\label{R_n}
    \lim_{n \to \infty} P\bigg\{ \sup_{\theta \in N_0(\delta)} |R_n(\theta)| <\epsilon \bigg\} = 1
.\end{equation}

Let's try to understand the behavior of the following integral (for $0<\epsilon <1$):
\begin{equation}\label{int_epsilon}
    I_{\epsilon} = \int_{N_0(\delta)}\exp(\lest - \ltrue)\exp\bigg[\dfrac{-(\theta - \mle)^2(1\pm \epsilon)}{2\sigma_n^2} \bigg] d\theta
.\end{equation}
From \citet{walker}, we already know how to control the $\mle$ related term. Because of high dimensional setting, we have new term to deal with. Note that we can express the difference of log likelihood as:
\begin{align*}
    \lest - \ltrue &= \underbrace{\{ (\lest - \text{L}_n(\theta_0 \mid h(Z), \phi) - ( \ltrue - \lalltrue)\}}_{\text{depends on $\theta$}} \\
    &+ \underbrace{\{ \text{L}_n(\theta_0 \mid h(Z), \phi) - \lalltrue\}}_{\text{independent of } \theta}.
\end{align*}
For the terms that depend on $\theta$, we can use Taylor's series approximation to further simplify them.

For every small $\epsilon >0,$ we can find $\delta>0$ using Lemma~\ref{lemma_hessians} and Lemma~\ref{lemma_scores} so that for $\theta$ in $N_0(\delta)$:
\begin{align}\label{expansion_any_nuisance}
    \lest - \text{L}_n(\theta_0 \mid h(Z), \phi) &= (\theta - \theta_0)(X - h(Z))^T(Y - \mu(\theta_0, h(Z), \phi)) \nonumber\\
    &- \frac{1}{2}(\theta - \theta_0)^2(X - h(Z))^T\Sigma(\theta_1, h(Z), \phi)(X - h(Z)) \nonumber\\
    &= (\theta - \theta_0)(X - h(Z))^T(Y - \mu(\theta_0, h(Z), \phi))\\
    &- \dfrac{(1\pm \epsilon)}{2}(\theta - \theta_0)^2(X - h(Z))^T\Sigma(\theta_0, h(Z), \phi)(X - h(Z)) \nonumber.
\end{align}

\begin{align}\label{expansion_true_nuisance}
    \ltrue - \lalltrue &= (\theta - \theta_0)(X - h_0(Z))^T(Y - \mu(\theta_0, h_0(Z), \phi_0))\nonumber\\
    &- \frac{1}{2}(\theta - \theta_0)^2(X - h_0(Z))^T\Sigma(\theta_2, h_0(Z), \phi_0)(X - h_0(Z))\nonumber\\
    &= (\theta - \theta_0)(X - h_0(Z))^T(Y - \mu(\theta_0, h_0(Z), \phi_0))\\
    &- \frac{(1\pm \epsilon)}{2}(\theta - \theta_0)^2(X - h_0(Z))^T\Sigma(\theta_0, h_0(Z), \phi_0)(X - h_0(Z)) \nonumber.
\end{align}

In the above expressions, $\theta_1$ and $\theta_2$ belong to $(\theta_0, \theta)$. Combining these two we get:
\[    (\text{L}_n(\theta \mid h(Z), \phi)- \text{L}_n(\theta_0 \mid h(Z), \phi)) - (\text{L}_n(\theta \mid h_0(Z), \phi_0) - \text{L}_n(\theta_0 \mid h_0(Z), \phi_0)) 
= (\theta - \theta_0)B_n - (\theta - \theta_0)^2d_n,\]
where 
\[B_n =  (X - h(Z))^T(Y - \mu(\theta_0, h(Z), \phi)) - (X - h_0(Z))^T(Y - \mu(\theta_0, h_0(Z), \phi_0)) \text{ and }\]
\[d_n = \frac{(1 \pm \epsilon)(\pm \epsilon_n)}{2} (X - h_0(Z))^T\Sigma(\theta_0, h_0(Z), \phi_0)(X - h_0(Z)) =: (1\pm \epsilon)d_{0n}.\]
For simplicity of notations, write $k_n = (1\pm \epsilon)/2\sigma_n^2$. Consequently, the integrand of $I_{\epsilon}$ becomes:

\begin{align*}
    &\exp(\lest - \ltrue)\exp\bigg[\dfrac{-(\theta - \mle)^2(1\pm \epsilon)}{2\sigma_n^2} \bigg] \\
    &= \exp(\text{L}_n(\theta_0 \mid h(Z), \phi) - \text{L}_n(\theta_0 \mid h_0(Z), \phi_0)) \times
    \exp((\theta - \theta_0)B_n - (\theta - \theta_0)^2d_n - (\theta - \mle)^2k_n).
\end{align*}
By change of variables to $\theta = \theta - \theta_0$ and completing the squares w.r.t. $\theta$, we can obtain the expression of $I_{\epsilon}$ as:
\begin{align*}
    I_{\epsilon} &= T(\theta_0,\mle, \epsilon, \epsilon_n) \int_{-\delta}^{\delta} \exp\bigg\{ \dfrac{-1}{2p_n^2}(\theta - \hat{\mu}_n)^2 \bigg\} d\theta \\
    &= T(\theta_0,\mle, \epsilon, \epsilon_n) (2\pi)^{\frac{1}{2}}p_n\big[\Phi(p_n^{-1}(\delta - \hat{\mu}_n)) -  \Phi(p_n^{-1}(-\delta - \hat{\mu}_n))\big],
\end{align*}
where
\[ T(\theta_0,\mle, \epsilon, \epsilon_n) = \exp\Bigg(\text{L}_n(\theta_0 \mid h(Z), \phi) - \text{L}_n(\theta_0 \mid h_0(Z), \phi_0) - (\theta_0 - \mle)^2k_n + \dfrac{(B_n - 2(\theta_0 - \mle)k_n)^2}{4(d_n + k_n)}\Bigg),\]
\[ \hat{\mu}_n = \frac{B_n - 2(\theta_0 - \mle)k_n}{2(d_n + k_n)} \text{, and } p_n^{-2} = 2(d_n + k_n).\]

We make the following claims:
\begin{equation}\label{new_mean_var}
    \hat{\mu}_n \prob 0 \text{ and } \frac{p_n^2}{\sigma_n^2} \prob 1.\end{equation}
Based on \eqref{new_mean_var}, we can observe that the new inverse of the variance term $p_n^{-2}$ increases at the same rate as $\sigma_n^{-2}$ which is of order $n$. The new mean ($\hat{\mu}_n$) converges in probability to zero, therefore, as long as $P(|\hat{\mu}_n|<\delta/2) =1 $, we can conclude that  \[ \big[\Phi(p_n^{-1}(\delta - \hat{\mu}_n)) -  \Phi(p_n^{-1}(-\delta - \hat{\mu}_n))\big] \prob 1.\]
To prove \eqref{new_mean_var}, note that:
\begin{align*}
    \dfrac{p_n^{-2}}{\sigma_n^{-2}} &= \dfrac{2(d_n +k_n)}{\sigma_n^{-2}} = \dfrac{(1 \pm \epsilon)(\pm \epsilon_n)(X - h_0(Z))^T\Sigma(\theta_0, h_0(Z), \phi_0)(X - h_0(Z)) + \sigma_n^{-2}(1\pm \epsilon)}{\sigma_n^{-2}(1\pm \epsilon)} \\
&=1+ \dfrac{(\pm \epsilon_n)(X - h_0(Z))^T\Sigma(\theta_0, h_0(Z), \phi_0)(X - h_0(Z))}{(X - h_0(Z))^T\Sigma(\mle, h_0(Z), \phi_0)(X - h_0(Z))} \prob 1,
\end{align*}
where $\epsilon_n = O_p\Big(Ks\sqrt{\frac{\log(d \vee n)}{n}}\Big) \bigg).$
For convergence of $\hat{\mu}_n$, rewrite it in the following way:
\begin{equation}\label{new_mean_split}
    \hat{\mu}_n = \dfrac{\frac{1}{n}(B_n - E[B_n])}{2(d_n + k_n)/n} + (\mle - \theta_0)\bigg(\dfrac{k_n}{d_n + k_n}\bigg) + \dfrac{E[B_n/n]}{2(d_n + k_n)/n}
.\end{equation}
By weak law of large numbers, consistency of the MLE, and $(d_n + k_n)/n \to 1$, the first two terms converge in probability to 0. By Lemma~\ref{lemma_score_diff_expectation}, $E[B_n/n] \prob 0.$

% \begin{remark}\label{new_mean_rate}
% Moreover, note that if $\Delta_n = o_p(\sqrt{n})$ then $\Delta_n \hat{\mu}_n \prob 0$ as well because of stochastic boundedness (consequence of CLT) of first two terms in \eqref{new_mean_split} and Lemma~\ref{lemma_score_diff_expectation} for the last term. This will be useful if we want to ensure $P(|\hat{\mu}_n|<\delta/2) =1 $ holds. 
% \end{remark}

We will simplify the term $T(\theta_0,\mle, \epsilon, \epsilon_n)$ and split it into terms that are independent of $\epsilon$, terms which converge in probability to one and terms that are dependent on $\epsilon$ but stochastically bounded. From Lemma~\ref{lemma_score_diff_expectation}, recall that 
\begin{align*}
\frac{1}{n}B_n &= b_n(h(Z), \phi) - b_n(h_0(Z),\phi_0)\\
&=\frac{1}{n}\sum_{i=1}^{n}\bigg[(X_i-h(Z_i))(Y_i-\mu_i(\theta_0,h(Z), \phi)) - (X_i-h_0(Z_i))(Y_i-\mu_i(\theta_0,h_0(Z),\phi_0)) \bigg].
\end{align*}
We will use the following facts which are consequences of \eqref{new_mean_var}, Central Limit Theorem and Lemma~\ref{lemma_score_diff_expectation}:
\begin{align}\label{lemma7_clt_more}
    &n^{\frac{1}{2}}(\theta_0 - \mle) = O_p(1),\nonumber\\
    &n^{\frac{1}{2}}(E[B_n]/n) \prob 0,
    n^{\frac{1}{2}}(B_n/n - E[B_n]/n) \prob 0,\\
    &(d_n + k_n)/n \to 1, (d_n+k_n)/k_n \to 1 \nonumber.
\end{align}

Based on the fact that $O_p(1)o_p(1) = o_p(1)$ and adding and subtracting the $E[B_n]$ term, we can write
\begin{align*}
    &\log \Big[ T(\theta_0,\mle, \epsilon, \epsilon_n)\Big] = \underbrace{\text{L}_n(\theta_0 \mid h(Z), \phi) - \text{L}_n(\theta_0 \mid h_0(Z), \phi_0)}_{T_n(\theta_0)}\\ &+ \dfrac{((B_n - E[B_n]) + E[B_n] - 2(\theta_0 - \mle)k_n)^2}{4(d_n + k_n)}- (\theta_0 - \mle)^2k_n \\
    &= T_n(\theta_0) + \dfrac{(B_n - E[B_n])E[B_n]}{2(d_n+k_n)} - \dfrac{E[B_n](\theta_0 - \mle)k_n}{d_n+k_n} + \dfrac{(E[B_n])^2}{4(d_n+k_n)}\\
    &+ (k_n/n)(\sqrt{n}(\theta_0 - \mle))^2\bigg( \dfrac{k_n}{d_n+k_n} -1\bigg)+ \dfrac{(B_n - E[B_n])^2}{4(d_n + k_n)} - \dfrac{(B_n - E[B_n])(\theta_0 - \mle)k_n}{d_n + k_n} \\
    &= T_n(\theta_0) - \\ &\underbrace{\dfrac{n^{\frac{1}{2}}(B_n/n - E[B_n]/n)(n^{\frac{1}{2}}E[B_n]/n)}{2(d_n+k_n)/n} - \dfrac{n^{\frac{1}{2}}(E[B_n]/n)n^{\frac{1}{2}}(\theta_0 - \mle)}{(d_n+k_n)/k_n} + \dfrac{(n^{\frac{1}{2}}(E[B_n]/n))^2}{4(d_n+k_n)/n}}_{o_P(1)}\\
    &+ \underbrace{\dfrac{k_n}{n}(\sqrt{n}(\theta_0 - \mle))^2\bigg( \dfrac{k_n}{d_n+k_n} -1\bigg)}_{o_P(1)}\\
    &+ \underbrace{\dfrac{(n^{\frac{1}{2}}(B_n/n - E[B_n/n]))^2}{4(d_n + k_n)/n} -  \dfrac{n^{\frac{1}{2}}(B_n/n - E[B_n]/n)n^{\frac{1}{2}}(\theta_0 - \mle)k_n}{d_n + k_n}}_{o_P(1)}\\
    &=T_n(\theta_0) + o_P(1).
\end{align*}

Therefore, we can write $T(\theta_0, \mle,\epsilon, \epsilon_n)$ as
% \begin{equation}
%     T(\theta_0, \mle,\epsilon, \epsilon_n) = \exp(T_n(\theta_0, \mle))\exp(r_n) \exp(\pm \epsilon)^{s_n},
% \end{equation}

\begin{equation}
    T(\theta_0, \mle,\epsilon, \epsilon_n) = \exp(T_n(\theta_0))\exp(r_n),
\end{equation}
where $r_n \prob 0$ and so $\exp(r_n) \prob 1$.
%, $s_n$ is stochastically bounded and 
% \begin{equation}\label{scale_I2}
%     T_n(\theta_0, \mle) = \text{L}_n(\theta_0 \mid h(Z), \phi) - \text{L}_n(\theta_0 \mid h_0(Z), \phi_0) - (B_n-E[B_n])(\theta_0 - \mle) + \dfrac{(B_n- E[B_n])^2}{(d_{0n}+\sigma_n^{-2})}.
% \end{equation}

Going back to the expression of $I_{\epsilon}$, we obtain that

\begin{align}\label{I_epsilon_final}
\begin{split}
    \exp[-T_n(\theta_0)]I_{\epsilon} &= (2\pi)^{\frac{1}{2}}\sigma_n(1\pm \epsilon)^{\frac{-1}{2}}\times \\
    &\underbrace{\dfrac{p_n}{\sigma_n(1\pm \epsilon)^{\frac{-1}{2}}}\exp(r_n) \times
    \big[\Phi(p_n^{-1}(\delta - \hat{\mu}_n)) -  \Phi(p_n^{-1}(-\delta - \hat{\mu}_n))\big]}_{\prob 1} \\
    &= (2\pi)^{\frac{1}{2}}\sigma_n(1\pm \epsilon)^{\frac{-1}{2}}(1 + o_P(1)).
    \end{split}
\end{align}
Therefore, from \eqref{I_3}, \eqref{R_n} and \eqref{I_epsilon_final}, we can observe that
\begin{equation}\label{I_3_final}
    \lim_{n \to \infty}P\Bigg\{  (2\pi)^{\frac{1}{2}}\sigma_n(1+ \epsilon)^{\frac{-1}{2}} < \exp[-T_n(\theta_0)]I_3 <  (2\pi)^{\frac{1}{2}}\sigma_n(1- \epsilon)^{\frac{-1}{2}}  \Bigg\} = 1.
\end{equation}

We can choose $\delta>0$ such that \eqref{I_2_to_I_3},  \eqref{expansion_any_nuisance}, \eqref{expansion_true_nuisance}, and \eqref{I_3_final} hold together for arbitrarily small $\epsilon>0$ by taking minimum of these four $\delta$ values and obtain
%(recall $s_n$ are stochastically bounded):
\begin{align}\label{I_3_to_I_2_final}
    \lim_{n \to \infty}P\Bigg\{ (2\pi)^{\frac{1}{2}}\sigma_n(1- \eta)<   &\exp[-T_n(\theta_0)]\{\pi(\theta_0) p_n(Y\mid\hat{\theta}_0,h_0(Z),\phi_0)\}^{-1} I_2\\
    &< (2\pi)^{\frac{1}{2}}\sigma_n(1+ \eta) \Bigg\} = 1 \nonumber,
\end{align}
for some arbitrarily small $\eta>0.$ This is equivalent to
\begin{equation}\label{I_2_behavior}
    \exp[-T_n(\theta_0)]\{p_n(Y\mid\hat{\theta}_0,h_0(Z),\phi_0) \sigma_n\}^{-1} I_2 \prob (2\pi)^{\frac{1}{2}}\pi(\theta_0)
.\end{equation}
We revisit the scale term $\exp[-T_n(\theta_0, \mle)]$.
%We can write (from \eqref{scale_I2}):
\begin{align}\label{T_n_rate}
  n^{-1}(-T_n(\theta_0)) &= n^{-1}\underbrace{(\text{L}_n(\theta_0 \mid h_0(Z), \phi_0) - \text{L}_n(\theta_0 \mid h(Z), \phi))}_{>0 \text{ with high prob}} 
  %+n^{-1} \bigg(\underbrace{(B_n-E[B_n])(\theta_0 - \mle) - \dfrac{(B_n- E[B_n])^2}{(d_{0n}+\sigma_n^{-2})}}_{O_p(1)}\bigg)\\
  = O_p\bigg( \sqrt{\dfrac{\log(d \vee n)}{n}} \bigg).
\end{align}
\textbf{Behavior of $I_1$:}

For $I_1$, we will use the expression in \eqref{mle_split} and rewrite the integral as:

\begin{align}
    I_1 &= p_n(Y\mid\mle, h_0(Z), \phi_0)\exp(\lalltrue - \lmle) \times \\
    & \int_{\Theta - N_0(\delta)} \pi(\theta) \exp(\lest - \lalltrue) d\theta.\nonumber
\end{align}
    
From Lemma~\ref{lemma_ldiff_complement}, we know that $\forall \delta >0$, the integrand above is less than $\pi(\theta)\exp(-nk(\delta))$ for some $k(\delta)>0$.
Therefore, 
\[ I_1 \leq p_n(Y\mid\mle, h_0(Z), \phi_0)\exp(\underbrace{\lalltrue - \lmle}_{\leq 0 \text{ by the definition of} \mle})\exp(-nk(\delta)).\] 
Using Lemma~\ref{lemma_walker} and the fact that $k(\delta)>0$ , we obtain
\begin{align*}
    \exp[-T_n(\theta_0)]\sigma_n^{-1}\exp-nk(\delta)) &= \{-n^{-1}\text{L}_n^{''}(\mle)\}^{1/2} \times \dfrac{n^{1/2}}{\exp\bigg[n\Big(k(\delta) - n^{-1}(-T_n(\theta_0))\Big)\bigg]}\\
    &= o_P(1).
\end{align*} 
Consequently, we get
\begin{equation}\label{I_1_behavior}
    \exp[-T_n(\theta_0)]\{p_n(Y\mid\mle, h_0(Z), \phi_0) \sigma_n\}^{-1}I_1 \prob 0
.\end{equation}
Combining \eqref{I_2_behavior} and \eqref{I_1_behavior}, we obtain
\begin{align}\label{marginal_behavior}
        &\exp[-T_n(\theta_0)]\{p_n(Y\mid\hat{\theta}_0,h_0(Z),\phi_0) \sigma_n\}^{-1} I \prob (2\pi)^{\frac{1}{2}}\pi(\theta_0),\\
\end{align}
where marginal $I$ is
\begin{align*}
\int_{\Theta}\pi(\theta)\exp(\text{L}_n(Y\mid X,Z,\theta, h(Z), \phi) d\theta.
\end{align*}

To complete the proof, we have to deal with the following integral
\[ I^{(a,b)} = \int_{\hat{\theta}_0 + a\sigma_n}^{\hat{\theta}_0 + b\sigma_n} \pi(\theta)\exp(\text{L}_n(Y\mid X,Z,\theta, h(Z), \phi) d\theta.\]
This is similar to the integral $I_2$ \eqref{I_split}, except the limits of integral are $(\hat{\theta}_0 + b\sigma_n, \hat{\theta}_0 + a\sigma_n)$ instead of $(\theta_0-\delta, \theta_0 + \delta)$. Consistency of Maximum Likelihood Estimate implies
\[ \lim_{n \to \infty} P\{(\hat{\theta}_0 + a\sigma_n, \hat{\theta}_0 + b\sigma_n) \subset (\theta_0-\delta, \theta_0 + \delta) \} = 1.\]
Instead of \eqref{I_epsilon_final}, we will have following:

\begin{align*}\label{Iab_epsilon_final}
    \exp[-T_n(\theta_0)]I_{\epsilon}^{(a,b)} &= (2\pi)^{\frac{1}{2}}\sigma_n(1\pm \epsilon)^{\frac{-1}{2}} \dfrac{p_n}{\sigma_n(1\pm \epsilon)^{\frac{-1}{2}}}\exp(r_n)\times\\
    &
    \big[\Phi(p_n^{-1}(\mle - \theta_0  - \hat{\mu}_n + a\sigma_n)) -  \Phi(p_n^{-1}(\mle - \theta_0 - \hat{\mu}_n +b\sigma_n))\big].
\end{align*}

From the expression of $\hat{\mu}_n$ \eqref{new_mean_split}, we have:
\begin{align*}
    p_n^{-1}(\mle - \theta_0 - \hat{\mu}_n +b\sigma_n) &= \Big(\dfrac{p_n^{-1}}{\sqrt{n}}\Big)\sqrt{n}(\mle - \theta_0)\bigg( 1- \dfrac{k_n}{d_n+k_n} \bigg) \\
    &- \Big(\dfrac{p_n^{-1}}{\sqrt{n}}\Big) \dfrac{\sqrt{n}(B_n/n - E[B_n]/n)}{2(d_n + k_n)/n}\\
    &- \Big(\dfrac{p_n^{-1}}{\sqrt{n}}\Big)\dfrac{\sqrt{n}E[B_n/n]}{2(d_n + k_n)/n} + b\Big(\dfrac{\sigma_n}{p_n}\Big)\\
    &\prob b(1\pm \epsilon)^{\frac{1}{2}}.
\end{align*}
Therefore, we have
\begin{equation*}
    \big[\Phi(p_n^{-1}(\mle - \theta_0  - \hat{\mu}_n + a\sigma_n)) -  \Phi(p_n^{-1}(\mle - \theta_0 - \hat{\mu}_n +b\sigma_n))\big] \prob [\Phi(a(1\pm \epsilon)^{\frac{1}{2}}) - \Phi(b(1\pm \epsilon)^{\frac{1}{2}})].
\end{equation*}

Using continuity of $\Phi$, we obtain following instead of \eqref{I_3_to_I_2_final}:
\begin{align}
    \lim_{n \to \infty}P\Bigg\{ (2\pi)^{\frac{1}{2}}\sigma_n(1- \eta)  &<\dfrac{\exp[-T_n(\theta_0)]\{\pi(\theta_0) p_n(Y\mid\hat{\theta}_0,h_0(Z),\phi_0)\}^{-1} }{(\Phi(a) - \Phi(b))}I^{(a,b)} \\
    &< (2\pi)^{\frac{1}{2}}\sigma_n(1+ \eta)  \Bigg\} = 1 \nonumber,
\end{align}
for some arbitrarily small $\eta>0.$ This is equivalent to 
\begin{equation}\label{I_ab_behavior}
     \exp[-T_n(\theta_0)]\{p_n(Y\mid\hat{\theta}_0,h_0(Z),\phi_0) \sigma_n\}^{-1} I^{(a,b)} \prob (2\pi)^{\frac{1}{2}}\pi(\theta_0)(\Phi(a) - \Phi(b)).
\end{equation}

Combining \eqref{marginal_behavior} and \eqref{I_ab_behavior} on set $B$, we obtain 

\begin{equation}\label{result_gamma_phi}
    P \Bigg\{\bigg|\int_{\hat{\theta}_0 + a\sigma_n}^{\hat{\theta}_0 + b\sigma_n} f(\theta \mid Y,X,Z,h(Z), \phi) - \frac{1}{\sqrt{2\pi}} \int_{a}^{b} \exp\bigg(-\frac{1}{2}z^2 dz \bigg) \bigg| \geq \epsilon \Bigg| Y,X,Z\Bigg\} \xrightarrow{n \to \infty} 0,
\end{equation}

where 
\begin{equation*}
    f(\theta \mid Y,X,Z, h(Z), \phi) = \frac{\pi(\theta)\exp\{L_n(\theta\mid h(Z), \phi)\}}{\int_{\Theta}\pi(\theta)\exp(L_n(\theta \mid h(Z), \phi)) d\theta}.
\end{equation*}

Based on \eqref{result_gamma_phi} and the arguments in \eqref{prob_diff}, we can obtain the final result
\begin{equation*}
     P \Bigg\{\bigg|\int_{\hat{\theta}_0 + a\sigma_n}^{\hat{\theta}_0 + b\sigma_n} f(\theta \mid Y,X,Z) - \frac{1}{\sqrt{2\pi}} \int_{a}^{b} \exp\bigg(-\frac{1}{2}z^2 dz \bigg) \bigg| \geq \epsilon \Bigg| Y,X,Z\Bigg\} \xrightarrow{n \to \infty} 0.
\end{equation*}
\end{proof}

\begin{proof}[Proof of Corollary 1]
Consider the set $A = (-\infty, \hat{q}_{\alpha/2}^{CB}]$. According to Theorem~\ref{theorem1}, with $a = -\infty, b= (\hat{q}_{\alpha/2}^{CB} - \hat{\theta}_0)/\sigma_n$, we obtain that  :
\begin{equation*}
    \Bigg| \alpha/2 - \Phi\bigg(\frac{\hat{q}_{\alpha/2}^{CB} - \hat{\theta}_0}{\sigma_n}\bigg) \Bigg| = o_P(1).
\end{equation*}
The above expression implies that 
\begin{equation*}
    \hat{q}_{\alpha/2}^{CB} = \hat{\theta}_0 + \sigma_n z_{\alpha/2} + o_P(1).
\end{equation*}
Let $\hat{q}_{\alpha/2}$ be the frequentist lower bound for the confidence interval, then
\begin{equation*}
    \hat{q}_{\alpha/2} = \hat{\theta}_0 + \sigma_n z_{\alpha/2}.
\end{equation*}
Therefore, we have
\begin{equation}\label{alpha/2}
    |\hat{q}_{\alpha/2}^{CB} - \hat{q}_{\alpha/2}| = o_P(1).
\end{equation}

Similarly, by considering the set $B = [\hat{q}_{1 - \alpha/2}^{CB},\infty)$, we can obtain

\begin{equation}\label{1-alpha/2}
    |\hat{q}_{1 - \alpha/2}^{CB} - \hat{q}_{1 - \alpha/2}| = o_P(1),
\end{equation}
where $\hat{q}_{1 - \alpha/2} = \hat{\theta}_0 + \sigma_n z_{1 - \alpha/2}$.
Using \eqref{alpha/2}, \eqref{1-alpha/2}, and the asymptotic normality of the maximum likelihood estimator $\hat{\theta}_0$, we obtain that for every $\epsilon, \delta$ positive, there exists, $N_{\epsilon,\delta}$ such that for all $n > N_{\epsilon,\delta}$,
\begin{align*}
    P(|P(\theta_0 \in (\hat{q}_{\alpha/2}^{CB},\hat{q}_{1 - \alpha/2}^{CB})\mid Y,X,Z) - (1-\alpha)| < \delta ) > 1 - \epsilon,
\end{align*}

or equivalently,
\begin{equation*}
    \Bigg| P \Big[\theta_0 \in \big(\hat{q}_{\alpha/2}^{CB},\hat{q}_{1 -\alpha/2}^{CB}\big) \mid Y, X, Z\Big] - (1-\alpha) \Bigg| \prob 0.
\end{equation*}
\end{proof}

% the below lines are only needed if bibliography precedes appendices
% uses https://tex.stackexchange.com/a/440212 to continue page numbering
% \clearpage
% \setcounter{counterforappendices}{\value{page}}
% \mainmatter
% \setcounter{page}{\value{counterforappendices}}

\end{document}